\documentclass{article}
\usepackage[active]{srcltx}
\usepackage{amsmath}
\usepackage{bbm}
\usepackage{amsthm}
\usepackage{amsfonts}
\usepackage{leftidx}
\usepackage{mathrsfs}
\usepackage{amssymb}
\usepackage{amscd}
\usepackage{cases}
\usepackage{a4wide}
\usepackage{bm}
\usepackage{latexsym}
\usepackage[usenames,dvipsnames]{xcolor}
\usepackage{dsfont}
\usepackage{indentfirst}
\usepackage[latin1]{inputenc}

\def\qed{\hfill $\square$}

\newtheorem{definition}{Definition}[section]
\newtheorem{proposition}[definition]{Proposition}
\newtheorem{theorem}[definition]{Theorem}
\newtheorem{lema}[definition]{Lemma}

\theoremstyle{definition}
\newtheorem{remark}[definition]{Remark}

\numberwithin{equation}{section}

\DeclareMathAlphabet{\mathpzc}{OT1}{pzc}{m}{it}

\begin{document}
\begin{center}
{\Large{ \textbf{Gaussian decay for a difference of traces of the Schr\"odinger semigroup associated to the isotropic harmonic oscillator.}}}

\medskip

\today

\end{center}

\begin{center}
\small{Mathieu Beau\footnote{Dublin Institute for Advanced Studies
School of Theoretical Physics, 10 Burlington Road, Dublin 4, Ireland; e-mail:
    mbeau@stp.dias.ie}, Baptiste Savoie\footnote{Department of Mathematical Sciences, Aarhus University, Ny Munkegade, Building 1530, DK-8000 Aarhus C, Denmark; e-mail: baptiste.savoie@gmail.com .}}
\end{center}

\begin{abstract}
This paper deals with the derivation of a sharp estimate on the difference of traces of the one-parameter Schr\"odinger semigroup associated to the quantum isotropic harmonic oscillator. Denoting by $H_{\infty,\kappa}$ the self-adjoint realization in $L^{2}(\mathbb{R}^{d})$, $d \in \{1,2,3\}$ of the Schr\"odinger operator $-\frac{1}{2} \Delta + \frac{1}{2} \kappa^{2}\vert \bold{x}\vert^{2}$, $\kappa>0$ and by $H_{L,\kappa}$, $L>0$ the Dirichlet realization in $L^{2}(\Lambda_{L}^{d})$ where $\Lambda_{L}^{d}:=\{\bold{x} \in \mathbb{R}^{d}:- \frac{L}{2} < x_{l} < \frac{L}{2},\,l=1,\ldots,d\}$, we prove that the difference of traces $\mathrm{Tr}_{L^{2}(\mathbb{R}^{d})} \mathrm{e}^{-t H_{\infty,\kappa}} - \mathrm{Tr}_{L^{2}(\Lambda_{L}^{d})}\mathrm{e}^{-t H_{L,\kappa}}$, $t>0$ has for $L$ sufficiently large a Gaussian decay in $L$. Furthermore, the estimate that we derive is sharp in the two following senses: its behavior when $t \downarrow 0$ is similar to the one given by $\mathrm{Tr}_{L^{2}(\mathbb{R}^{d})}\mathrm{e}^{-t H_{\infty,\kappa}} = (2\sinh( \frac{\kappa}{2}t))^{-d}$ and the exponential decay in $t$ arising from $\mathrm{Tr}_{L^{2}(\mathbb{R}^{d})}\mathrm{e}^{-t H_{\infty,\kappa}}$ when $t\uparrow \infty$ is preserved. For illustrative purposes, we give a simple application within the framework of quantum statistical mechanics.
\end{abstract}

\medskip

\noindent
\textbf{MSC-2010 number}:  35J10, 47D08, 81Q10, 81Q15, 82B10.

\medskip

\medskip


\section{Introduction and the main result.}

\subsection{Introduction.}

Within the framework of statistical mechanics, the thermodynamic description of large macroscopic equilibrium systems is obtained by considering the thermodynamic (i.e. infinite volume) limit of finite systems described by suitable Gibbs ensembles (e.g. micro-canonical, canonical, grand-canonical). Whenever the thermodynamic limit exists, the surface effects disappear and one is left with the \textit{bulk properties}. If the limit only depends on the intensive parameters  then the system has the extensive property, i.e. the thermodynamic quantities are asymptotically proportional to the system size. Investigating the existence of the thermodynamic limit is part of fundamental topics addressed in rigorous statistical mechanics. It is usually a non-trivial problem which is even more difficult that the interactions involved are singular. The existence of the thermodynamic limit depends on the nature of the interactions, and
therefore, may not exist. There is a wide literature on this topic, we refer to \cite{Ru} for an overview. Whenever the limit exists, one has to prove in addition that it is independent of the boundary conditions and of the sequence of confining boxes considered. 
A related topic consists in estimating the boundary contributions. This allows one to answer the question: how negligible is the error made by replacing the finite-volume statistical quantities with the corresponding thermodynamic limit?\\
\indent One of the most widespread models in physics is the isotropic quantum harmonic oscillator. Turning to the statistical description of $d$-dimensionnal ($d \in \{1,2,3\}$) harmonically trapped gases, the finite-volume quantities in the grand-canonical ensemble can be written under certain conditions  in terms of the trace of the one-parameter 'finite-volume' semigroup (the 'inverse temperature' plays the role of parameter) generated by the Hamiltonian \eqref{HL} with Dirichlet boundary conditions. Without loss of generality, we here consider cubic confining boxes. For such systems, the existence of the conventional thermodynamic limit is trivial since the one-parameter 'infinite-volume' semigroup generated by the Hamiltonian \eqref{Hinfini} on the whole space is trace class. This stems from the discrete nature of the spectrum of \eqref{Hinfini} along with the explicit knowledge of its eigenvalues, see \eqref{vap1}-\eqref{vapd} below. The error made by replacing the trace of the finite-volume semigroup with the trace of the corresponding infinite-volume semigroup is expected to be negligible on the basis of an heuristical argument: the confinement arising from the harmonic potential overwhelms the confining effects of the 'walls' for sufficiently large boxes. From the foregoing, a natural question arises: how small is this error?\\
\indent The purpose of this paper is to estimate the remainder term in the asymptotic expansion of the trace of the finite-volume semigroup generated by the Hamiltonian in \eqref{HL} in the large-volume limit. We prove that for any dimension $d \in \{1,2,3\}$ the remainder term has a Gaussian decay in the characteristic length $L$ of the cubic confining box for $L$ sufficiently large, see Theorem \ref{Mimp} below. To our best knowledge, there is no proof of such a result in literature. A series of remarks placed below Theorem \ref{Mimp} discuss how the estimate is sharp. Its proof relies on a geometric perturbation theory applied to bounded operators on separable Hilbert spaces. Such a method has been originally developed to treat approximations of resolvent operators in \cite{CN1}, see also \cite{CFFH,S1} for further applications. Our paper extends the method to treat approximations of semigroups.\\
\indent For the sake of completeness, we mention that an asymptotic expansion in the high-temperature regime of the trace of the finite-volume semigroup (generated by operators of type \eqref{HL}) was derived in \cite[Thm. 3]{vdb}. The method is essentially based on the Feynman-Kac formula and Wiener estimates.  However, no information on the large-volume behavior can be inferred from. Note that \cite[Thm. 3]{vdb} was applied to investigate the thermodynamics of quantum gases confined by weak external potentials, see \cite{vdbl} for further details.

\subsection{The setting and the main result.}
\label{defop}

For any $d \in \{1,2,3\}$ and $L \in (0,\infty)$, denote $\Lambda_{L}^{d}:= \{\bold{x} \in \mathbb{R}^{d}: -\frac{L}{2} < x_{l} < \frac{L}{2},\, l= 1,\ldots,d\}$  and $\vert \Lambda_{L}^{d}\vert$ its Lebesgue-measure. On $\mathcal{C}_{0}^{\infty}(\Lambda_{L}^{d})$, define $\forall \kappa>0$ the family of operators:
\begin{equation}
\label{HL}
H_{L,\kappa} := \frac{1}{2}\left(- i \nabla_{\bold{x}}\right)^{2} + \frac{1}{2} \kappa^{2} \vert \bold{x}\vert^{2}.
\end{equation}
By standard arguments, $\forall \kappa>0$ \eqref{HL} extends to a family of self-adjoint and bounded from below operators $\forall L \in (0,\infty)$, denoted again by $H_{L,\kappa}$, with domain $D(H_{L,\kappa}) = W_{0}^{1,2}(\Lambda_{L}^{d}) \cap W^{2,2}(\Lambda_{L}^{d})$. Obviously this definition corresponds to choose Dirichlet boundary conditions on the boundary $\partial \Lambda_{L}^{d}$. Since the inclusion $W_{0}^{1,2}(\Lambda_{L}^{d}) \hookrightarrow L^{2}(\Lambda_{L}^{d})$ is compact, then $\forall \kappa>0$ $H_{L,\kappa}$ has a purely discrete spectrum with an accumulation point at infinity.\\
\indent When $\Lambda_{L}^{d}$ fills the whole space (i.e. when $L \uparrow \infty$), define $\forall \kappa>0$ on $\mathcal{C}_{0}^{\infty}(\mathbb{R}^{d})$ the operator:
\begin{equation}
\label{Hinfini}
H_{\infty,\kappa}:= \frac{1}{2}\left(- i \nabla_{\bold{x}}\right)^{2} + \frac{1}{2} \kappa^{2} \vert \bold{x}\vert^{2}.
\end{equation}
From \cite[Thm. X.28]{RS2}, $\forall \kappa>0$ \eqref{Hinfini} is essentially self-adjoint and its self-adjoint extension, denoted again by $H_{\infty,\kappa}$, is semi-bounded. By \cite[Thm. XIII.16]{RS4}, the spectrum of $H_{\infty,\kappa}$ is purely discrete with eigenvalues increasing to infinity. From the one-dimensional problem, the eigenvalues and eigenfunctions of the multidimensional case can be written down explicitly. The eigenvalues of the one-dimensional problem are all simple (i.e. non-degenerate) and given by, see e.g. \cite[Sec. 1.8]{BS}:
\begin{equation}
\label{vap1}
\epsilon_{\infty,\kappa}^{(s)} := \kappa \left(s + \frac{1}{2}\right),\quad s \in \mathbb{N}.
\end{equation}
The corresponding eigenfunctions which form an orthonormal basis in $L^{2}(\mathbb{R})$, read as:
\begin{equation}
\label{funcp1}
\forall x \in \mathbb{R},\quad \phi_{\infty,\kappa}^{(s)}(x) := \frac{1}{\sqrt{2^{s} s!}} \left(\frac{\kappa}{\pi}\right)^{\frac{1}{4}} \mathrm{e}^{- \frac{\kappa}{2} x^{2}} \mathpzc{H}_{s}\left(\sqrt{\kappa} x\right), \quad s \in \mathbb{N},
\end{equation}
where $\mathpzc{H}_{s}$, $s \in \mathbb{N}$ are the Hermite polynomials defined by: $\mathpzc{H}_{s}(x) := (-1)^{s} \mathrm{e}^{x^{2}} \frac{\mathrm{d}^{s}}{\mathrm{d} x^{s}}(\mathrm{e}^{-x^{2}})$, $\forall x \in \mathbb{R}$.\\
The eigenvalues and eigenfunctions of the  multidimensional case (i.e. $d=2,3$) are respectively related to those of the one-dimensional case by:
\begin{gather}
\label{vapd}
E_{\infty,\kappa}^{(\bold{s})} := \sum_{j=1}^{d} \epsilon_{\infty,\kappa}^{(s_{j})} = \kappa \sum_{j=1}^{d} \left(s_{j} + \frac{1}{2}\right),\quad \bold{s}=\left\{s_{j}\right\}_{j=1}^{d} \in \mathbb{N}^{d},\\
\label{funcpd}
\psi_{\infty,\kappa}^{(\bold{s})}(\bold{x}) := \prod_{j=1}^{d} \phi_{\infty,\kappa}^{(s_{j})}(x_{j}),\quad \bold{x}= \left\{x_{j}\right\}_{j=1}^{d} \in \mathbb{R}^{d}.
\end{gather}
From \eqref{vap1}-\eqref{vapd} and by the use of the min-max principle, one has for any $L \in (0,\infty)$:
\begin{equation*}
\forall \kappa>0,\quad \inf \sigma\left(H_{L,\kappa}\right) \geq \inf \sigma \left(H_{\infty,\kappa}\right) = E_{\infty,\kappa}^{(\bold{0})} = d \epsilon_{\infty,\kappa}^{(0)}>0,\quad \epsilon_{\infty,\kappa}^{(0)}:= \frac{\kappa}{2}.
\end{equation*}

Let us turn to the one-parameter strongly-continuous semigroup (the so-called $\mathrm{C}_{0}$-semigroup in the Hille-Phillips terminology \cite{HP}) generated by the operators introduced above. At finite-volume, it is defined $\forall L\in (0,\infty)$ and $\forall \kappa>0$ by $\{G_{L,\kappa}(t) := \mathrm{e}^{- t H_{L,\kappa}} : L^{2}(\Lambda_{L}^{d}) \rightarrow L^{2}(\Lambda_{L}^{d})\}_{t \geq 0}$. It is a self-adjoint and positive operator on $L^{2}(\Lambda_{L}^{d})$ by the spectral theorem and the functional calculus, see e.g. \cite[Sec. X.8]{RS2}. The same hold true for the one-parameter semigroup on the whole space $\{G_{\infty,\kappa}(t) := \mathrm{e}^{- t H_{\infty,\kappa}} : L^{2}(\mathbb{R}^{d}) \rightarrow L^{2}(\mathbb{R}^{d})\}_{t \geq 0}$. Moreover $\forall 0<L\leq \infty$, $\forall \kappa>0$ and $\forall t>0$, $G_{L,\kappa}(t)$ is a Gibbs semigroup, i.e. $G_{L,\kappa}(t)$ (resp. $G_{\infty,\kappa}(t)$) belongs to the Banach space of trace-class operators on $L^{2}(\Lambda_{L}^{d})$ (resp.
$L^{2}(\mathbb{R}^{d})$), see e.g. \cite{U,ABN}. A basic feature is the monotonicity property for the finite-volume trace, see Lemma \ref{2trace} in Sec. \ref{Appen1}:
\begin{equation*}
\forall L \in (0,\infty),\quad \mathrm{Tr}_{L^{2}(\Lambda_{L}^{d})}\left\{G_{L,\kappa}(t)\right\} \leq \mathrm{Tr}_{L^{2}(\mathbb{R}^{d})}\left\{G_{\infty,\kappa}(t)\right\} = \left(2\sinh\left(\frac{\kappa}{2} t\right)\right)^{-d},\quad \kappa>0,\, t>0.
\end{equation*}

Our main result is the following sharp estimate on the difference of traces of the semigroups:

\begin{theorem}
\label{Mimp}
For any $d \in \{1,2,3\}$ there exists a constant $C_{d}>0$ and $\forall 0<\kappa_{0}<1$ there exists a real
$\mathpzc{L}_{\kappa_{0}}>0$ s.t. $\forall L \in [\mathpzc{L}_{\kappa_{0}},\infty)$, $\forall \kappa \in [\kappa_{0},\infty)$ and $\forall t>0$:
\begin{multline}
\label{cherc}
\left\vert  \mathrm{Tr}_{L^{2}(\Lambda_{L}^{d})}\left\{G_{L,\kappa}(t)\right\} - \mathrm{Tr}_{L^{2}(\mathbb{R}^{d})}\left\{G_{\infty,\kappa}(t)\right\}\right\vert \\ \leq C_{d} \left(1+\sqrt{\kappa}\right)(1+\kappa)^{d} (1+t)^{3(d+ \frac{1}{2})}  \mathrm{Tr}_{L^{2}(\mathbb{R}^{d})}\left\{G_{\infty,\kappa}(t)\right\} \mathrm{e}^{-\frac{\kappa}{32} \frac{L^{2}}{4} \tanh\left(\frac{\kappa}{2}t\right)}.
\end{multline}
\end{theorem}

\begin{remark}
The estimate is sharp in the following senses. Its behavior when $t \downarrow 0$ is given by the term $\mathrm{Tr}_{L^{2}(\mathbb{R}^{d})}\{G_{\infty,\kappa}(t)\}= (2\sinh(\frac{\kappa}{2} t))^{-d}$. We recall that $\sinh(x) \sim x$ when $x \downarrow 0$. Moreover, the r.h.s. of \eqref{cherc} has an exponential decay in $t$ when $t \uparrow \infty$ arising from $\mathrm{Tr}_{L^{2}(\mathbb{R}^{d})}\left\{G_{\infty,\kappa}(t)\right\}$. Note that one can get rid of the polynomial growth in $t$ appearing in the r.h.s. of \eqref{cherc} via \eqref{redexp}.
\end{remark}

\begin{remark}
The $\mathpzc{L}$ in Theorem \ref{Mimp} is independent of $\kappa$ whenever $\kappa \in [1,\infty)$. If $0<\kappa<1$, then the estimate holds for $L$ large enough chosen accordingly (i.e. $L \geq cste \times \kappa^{-\frac{1}{2}}$).
\end{remark}

\begin{remark}
In \eqref{cherc}, the powers on the factors $(1+\sqrt{\kappa})$, $(1+\kappa)$, $(1+t)$ and the constant appearing in the argument of the exponential can be optimized.
\end{remark}

\subsection{An application in quantum statistical mechanics.}

Consider a $d$-dimensional ideal quantum gas composed of a large number of non-relativistic spin-0 identical particles confined in the box $\Lambda_{L}^{d}$ and trapped in an isotropic harmonic potential. Such a system is considered to figure out the Bose-Einstein condensation phenomenon created by cold alkali atom gases in magnetic-optical trap, see e.g. \cite[Chap. 10]{Phys3} and references therein. Within the one-body approximation, the dynamics of a single Boson is determined by \eqref{HL}. Suppose that the system is at equilibrium with a thermal and particles bath. In the grand-canonical ensemble, let $(\beta,z,\vert \Lambda_{L}^{d}\vert)$ be the external parameters. Here, $\beta := (k_{B} T)^{-1}>0$ is the 'inverse' temperature ($k_{B}$ stands for the Boltzmann constant) and $z=\mathrm{e}^{\beta \mu}$ the fugacity ($\mu$ is the chemical potential). The finite-volume single-particle partition function is defined as, see e.g. \cite{Ru}:
\begin{equation}
\label{PhiL}
\Phi_{L,\kappa}(\beta) := \mathrm{Tr}_{L^{2}(\Lambda_{L}^{d})}\left\{G_{L,\kappa}(\beta)\right\},\quad \beta>0.
\end{equation}
The grand-canonical average number of particles at finite-volume is related to \eqref{PhiL} by, see \cite{BSa}:
\begin{equation}
\label{Navre}
\overline{N}_{L,\kappa}(\beta,z) := \sum_{l=1}^{\infty} z^{l} \Phi_{L,\kappa}(l \beta),\quad \beta>0,\, z \in \left(0,\mathrm{e}^{\beta \inf \sigma(H_{L,\kappa})}\right).
\end{equation}
Theorem \ref{Mimp} allows one to get the \textit{large-volume behavior} of the quantities in \eqref{PhiL}-\eqref{Navre}. Indeed, one gets $\forall 0<\kappa_{1}<\kappa_{2}<\infty$, $\forall 0<\beta_{1}<\beta_{2}<\infty$ and for any compact subset $K \subset (0, \mathrm{e}^{\beta_{1} E_{\infty,\kappa_{1}}^{(\bold{0})}})$:
\begin{gather*}
\Phi_{\infty,\kappa}(\beta):= \lim_{L \uparrow \infty} \Phi_{L,\kappa}(\beta) = \mathrm{Tr}_{L^{2}(\mathbb{R}^{d})}\left\{G_{\infty,\kappa}(\beta)\right\} = \frac{\mathrm{e}^{-\beta E_{\infty,\kappa}^{(\bold{0})}}}{\left(1 - \mathrm{e}^{-\beta \kappa}\right)^{d}},\\
\overline{N}_{\infty,\kappa}(\beta,z) := \lim_{L \uparrow \infty} \overline{N}_{L,\kappa}(\beta,z) = \sum_{l=1}^{\infty} z^{l} \Phi_{\infty,\kappa}(l \beta),
\end{gather*}
uniformly in $(\kappa,\beta,z) \in [\kappa_{1},\kappa_{2}]\times [\beta_{1},\beta_{2}]\times K$. Moreover, one has the following asymptotics:
\begin{gather*}
\Phi_{L,\kappa}(\beta) = \Phi_{\infty,\kappa}(\beta) + \mathcal{O}\left(\mathrm{e}^{-c L^{2}}\right),\\
\overline{N}_{L,\kappa}(\beta,z) = \overline{N}_{\infty,\kappa}(\beta,z) + \mathcal{O}\left(\mathrm{e}^{-c L^{2}}\right),
\end{gather*}
for some $L$-independent constant $c=c(\kappa,\beta)>0$. We emphasize that the upper bound in \eqref{cherc} plays a crucial to prove the thermodynamic limit of \eqref{Navre} for any $z \in (0,\mathrm{e}^{\beta E_{\infty,\kappa}^{(\bold{0})}})$, see \cite[Sec. A]{BSa}.

\section{Proof of Theorem \ref{Mimp}.}

The starting-point consists in rewriting the difference between the traces involving the difference between the semigroup integral kernels. We refer the reader to Sec. \ref{Appen1} in which we have collected some basic properties on the semigroup kernel. Since $\forall L\in (0,\infty]$ and $\forall \kappa>0$, $\{G_{L,\kappa}(t)\}_{t>0}$ is a Gibbs semigroup with a jointly continuous integral kernel $G_{L,\kappa}^{(d)}(\cdot\,,\cdot\,;t): \mathbb{R}^{d}\times\mathbb{R}^{d}\rightarrow \mathbb{C}$, then:
\begin{equation*}
\mathrm{Tr}_{L^{2}(\mathbb{R}^{d})}\left\{G_{\infty,\kappa}(t)\right\} - \mathrm{Tr}_{L^{2}(\Lambda_{L}^{d})}\left\{G_{L,\kappa}(t)\right\} = \mathscr{Y}_{L,\kappa}^{(d)}(t) + \mathscr{Z}_{L,\kappa}^{(d)}(t),
\end{equation*}
with $\forall d \in \{1,2,3\}$, $\forall L \in(0,\infty)$, $\forall \kappa>0$ and $\forall t>0$:
\begin{align}
\label{Yp}
\mathscr{Y}_{L,\kappa}^{(d)}(t) &:= \int_{\Lambda_{L}^{d}}  \left\{G_{\infty,\kappa}^{(d)}(\bold{x},\bold{x};t) - G_{L,\kappa}^{(d)}(\bold{x},\bold{x};t)\right\}\, \mathrm{d}\bold{x},\\
\label{Yr}
\mathscr{Z}_{L,\kappa}^{(d)}(t) &:= \int_{\mathbb{R}^{d} \setminus \Lambda_{L}^{d}} G_{\infty,\kappa}^{(d)}(\bold{x},\bold{x};t) \,\mathrm{d}\bold{x}.
\end{align}
Here, we used \cite[Prop. 9]{C1}. Note that $\forall \kappa>0$, $\forall t>0$ the kernel $G_{\infty,\kappa}^{(d)}(\cdot\,,\cdot\;t)$ is explicitly known and it is given by Mehler's formula, see \eqref{Mehler}-\eqref{multd}. It is derived from \eqref{funcp1}-\eqref{funcpd} and \eqref{vapd}.\\
\indent Next, it remains to estimate each one of the above quantity. For the quantity in \eqref{Yr}:
\begin{lema}
\label{lemth1}
$\forall d \in \{1,2,3\}$, $\forall L \in (0,\infty)$, $\forall \kappa>0$ and $\forall t>0$:
\begin{equation*}
\mathscr{Z}_{L,\kappa}^{(d)}(t) \leq \left(2 \sinh\left(\frac{\kappa}{2} t\right)\right)^{-d} \mathrm{e}^{- d \kappa \frac{L^{2}}{4} \tanh\left(\frac{\kappa}{2} t\right)}.
\end{equation*}
\end{lema}

\noindent \textit{Proof.} Let $\beta>0$ and $\kappa>0$ be fixed. Because of  \eqref{multd}, it is enough to treat only the case of $d=1$. From \eqref{Mehler} and by setting $x=y$, one has by direct computations:
\begin{equation*}
\forall L \in (0,\infty),\, \forall t>0,\quad \mathscr{Z}_{L,\kappa}^{(d=1)}(t) = \frac{\mathrm{erfc}\left(\sqrt{ \kappa \tanh\left(\frac{\kappa}{2} t\right)}\frac{L}{2}\right)}{\sqrt{2 \sinh(\kappa t)\tanh\left(\frac{\kappa}{2} t\right)}},
\end{equation*}
where $\mathrm{erfc}$ denotes the complementary error function, see e.g. \cite[Eq. (7.1.2)]{AS}. From the Chernoff inequality which reads as $\mathrm{erfc}(\alpha) \leq \mathrm{e}^{- \alpha^{2}}$ $\forall \alpha\geq 0$, along with the identity:
\begin{equation}
\label{Id1}
\sinh(\alpha t) = 2 \sinh\left(\frac{\alpha}{2} t\right) \cosh\left(\frac{\alpha}{2} t\right),\quad \forall \alpha\geq 0,\,\forall t>0,
\end{equation}
one arrives at:
\begin{equation*}
\forall L \in (0,\infty),\, \forall t>0,\quad \mathscr{Z}_{L,\kappa}^{(d=1)}(t) \leq \left(2 \sinh\left(\frac{\kappa}{2} t\right)\right)^{-1} \mathrm{e}^{-  \kappa  \frac{L^{2}}{4} \tanh\left(\frac{\kappa}{2} t\right)}.  \tag*{\qed}
\end{equation*}

As for the the quantity defined in \eqref{Yp}, we establish the following estimate:

\begin{proposition}
\label{proth1}
For any $d \in \{1,2,3\}$, there exists a constant $C_{d}>0$ and $\forall 0<\kappa_{0}<1$ there exists a real $\mathpzc{L}_{\kappa_{0}}>0$ s.t. $\forall L \in [\mathpzc{L}_{\kappa_{0}},\infty)$, $\forall \kappa \in [\kappa_{0},\infty)$ and $\forall t>0$:
\begin{equation}
\label{majYp}
\left\vert \mathscr{Y}_{L,\kappa}^{(d)}(t)\right\vert \leq C_{d} \left(1+\sqrt{\kappa}\right)(1+\kappa)^{d} (1+t)^{3(d+ \frac{1}{2})} \left(2\sinh\left(\frac{\kappa}{2} t\right)\right)^{-d} \mathrm{e}^{-\frac{\kappa}{32} \frac{L^{2}}{4} \tanh\left(\frac{\kappa}{2}t\right)}.
\end{equation}
\end{proposition}

Theorem \ref{Mimp} follows from Lemma \ref{lemth1} and Proposition \ref{proth1} together. The rest of this section is devoted to the proof of Proposition \ref{proth1}.

\subsection{Proof of Proposition \ref{proth1}.}

In view of \eqref{Yp}, the first step consists in writing an expression for the difference between the two semigroup kernels. It is contained in the following lemma:

\begin{lema}
\label{expdiff}
$\forall L \in (0,\infty)$, $\forall \kappa>0$ and $\forall t>0$:
\begin{multline}
\label{diff1}
\forall (x,y) \in \Lambda_{L}^{2},\quad G_{\infty,\kappa}^{(1)}(x,y;t) - G_{L,\kappa}^{(1)}(x,y;t) =  \\- \frac{1}{2} \int_{0}^{t}  \left\{G_{\infty,\kappa}^{(1)}(x,-\frac{L}{2};s) \left(\partial_{z} G_{L,\kappa}^{(1)}\right)(-\frac{L}{2},y;t-s) - G_{\infty,\kappa}^{(1)}(x,\frac{L}{2};s) \left(\partial_{z} G_{L,\kappa}^{(1)}\right)(\frac{L}{2},y;t-s)\right\}\, \mathrm{d}s,
\end{multline}
and in the case of $d=2,3$, for any $(\bold{x},\bold{y}) \in \Lambda_{L}^{2d}$:
\begin{equation}
\label{diff2}
G_{\infty,\kappa}^{(d)}(\bold{x},\bold{y};t) - G_{L,\kappa}^{(d)}(\bold{x},\bold{y};t) = - \frac{1}{2} \int_{0}^{t} \int_{\partial \Lambda_{L}^{d}} G_{\infty,\kappa}^{(d)}(\bold{x},\bold{z};s)\left[ \bold{n}_{\bold{z}} \cdot \nabla_{\bold{z}} G_{L,\kappa}^{(d)}(\bold{z},\bold{y};t-s)\right]\,   \mathrm{d}\sigma(\bold{z})\,\mathrm{d}s,
\end{equation}
where $\mathrm{d}\sigma(\bold{z})$ denotes the measure on $\partial \Lambda_{L}^{d}$ and $\bold{n}_{\bold{z}}$ the outer normal to $\partial \Lambda_{L}^{d}$ at $\bold{z}$.
\end{lema}

The proof of Lemma \ref{expdiff} in the case of $d=3$ can be found in \cite[Lem. 4.2]{BCL1}, see also \cite{BCL2}. Since the generalization to $d=1,2$ can be easily obtained by similar arguments, we do not give any proof.\\

\indent Recall that the kernel $G_{\infty,\kappa}^{(d)}$ is explicitly known and given by Mehler's formula. In view of \eqref{Yp} along with the expressions \eqref{diff1}-\eqref{diff2}, the actual problem comes down to derive a sufficiently sharp estimate on the gradient of the finite-volume semigroup kernel allowing us to bring out a Gaussian decay in $L$ for the quantity in \eqref{Yp}. It is contained in the following proposition:

\begin{proposition}
\label{cruspro}
For any $d \in \{1,2,3\}$, there exists a constant $C_{d}>0$ and $\forall 0<\kappa_{0}<1$ there exists a $\mathpzc{L}_{\kappa_{0}}>0$ s.t. $\forall L \in [\mathpzc{L}_{\kappa_{0}},\infty)$, $\forall \kappa \in [\kappa_{0},\infty)$, $\forall (\bold{x},\bold{y}) \in \Lambda_{L}^{2d}$ and $\forall t>0$:
\begin{gather}
\label{crus}
\left\vert \nabla_{\bold{x}} G_{L,\kappa}^{(d)}(\bold{x},\bold{y};t) \right\vert \leq C_{d} \left\{ \mathpzc{P}_{\infty,\kappa}^{(d)}(\bold{x},\bold{y};t) + \mathpzc{R}_{L,\kappa}^{(d)}(\bold{x},\bold{y};t)\right\},\\
\label{Q1e}
\mathpzc{P}_{\infty,\kappa}^{(d)}(\bold{x},\bold{y};t) := (1+\sqrt{\kappa})(1+t)^{\frac{5}{2}} \sqrt{\coth\left(\frac{\kappa}{2}t\right)}G_{\infty,\kappa}^{(d)}(\bold{x},\bold{y};t,8);\\
\label{R1}
\mathpzc{R}_{L,\kappa}^{(d)}(\bold{x},\bold{y};t) :=  \kappa^{\frac{d}{2}} (1+\kappa)^{\frac{d}{2}} \frac{t^{\frac{d-1}{2}}}{(\sinh(\kappa t))^{\frac{d}{2}}} (1+t)^{\frac{5d}{2}+1} \mathrm{e}^{-\frac{\kappa}{16} \frac{L^{2}}{4} \tanh\left(\frac{\kappa}{2} t\right)} G_{\infty,0}^{(d)}(\bold{x},\bold{y};4t).
\end{gather}
Here, $G_{\infty,\kappa}^{(d)}(\cdot\,,\cdot\,;t,\gamma)$, $\kappa>0$ and $\gamma>0$ is defined in \eqref{ginfty} and $G_{\infty,0}^{(d)}(\cdot\,,\cdot\,;t)$ in \eqref{heatk}-\eqref{multd}.
\end{proposition}

\begin{remark}
The $\mathpzc{L}$ in Proposition \ref{cruspro} can be chosen uniformly in $\kappa$ whenever $\kappa \geq 1$. If $0<\kappa<1$ then the estimate holds for $L$ large enough chosen accordingly (i.e. $L \geq cste/\sqrt{\kappa}$).
\end{remark}

Note that Proposition \ref{cruspro} actually contains the key-estimate of this paper; its proof is placed in  Sec. \ref{proofpro}. We mention that the derivation of such an estimate leans on a Duhamel-like formula for the finite-volume semigroup $G_{L,\kappa}(t)$, $L \in (0,\infty)$ obtained via a geometric perturbation theory.\\

\noindent \textit{Proof of Proposition \ref{proth1}.} Let $0 < \kappa_{0} < 1$ be fixed. Denote $\varsigma_{L} = \pm L/2$. Start with the case of $d=1$. In view of \eqref{Yp}, \eqref{diff1} and \eqref{crus}, we need to estimate $\forall L \in [\mathpzc{L}_{\kappa_{0}},\infty)$ and $\forall \kappa \in [\kappa_{0},\infty)$:
\begin{align}
\label{calZ1}
\forall t>0,\quad \mathpzc{Y}_{L,\kappa}^{(d=1),1}(t) &:= \frac{1}{2} \int_{0}^{t}   \int_{\Lambda_{L}^{1}}  G_{\infty,\kappa}^{(d=1)}(x,\varsigma_{L};s,1) \mathpzc{P}_{\infty,\kappa}^{(d=1)}(\varsigma_{L},x;t-s)\, \mathrm{d}x\,\mathrm{d}s,\\
\label{calZ2}
\mathpzc{Y}_{L,\kappa}^{(d=1),2}(t) &:= \frac{1}{2} \int_{0}^{t}  \int_{\Lambda_{L}^{1}}  G_{\infty,\kappa}^{(d=1)}(x,\varsigma_{L};s,1) \mathpzc{R}_{L,\kappa}^{(d=1)}(\varsigma_{L},x;t-s)\, \mathrm{d}x\,\mathrm{d}s.
\end{align}
Here, we have commuted the two integrals; this will be justified by what follows. We first estimate  \eqref{calZ1}. In view of \eqref{Mehler} and \eqref{Q1e}, then from \eqref{propsemi2} for any $L \in
[\mathpzc{L}_{\kappa_{0}},\infty)$, $\kappa \in [\kappa_{0},\infty)$ and $t>0$:
\begin{equation}
\label{excalY1}
\mathpzc{Y}_{L,\kappa}^{(d=1),1}(t) \leq C \sqrt{\kappa} \left(1+\sqrt{\kappa}\right) (1+t)^{\frac{5}{2}} \frac{1}{\sqrt{2\sinh(\kappa t)}}  \mathrm{e}^{- \frac{\kappa}{8} \frac{L^{2}}{4} \tanh\left(\frac{\kappa}{2}t\right)} \int_{0}^{t} \sqrt{\coth\left(\frac{\kappa}{2} (t-s)\right)}\,\mathrm{d}s,
\end{equation}
for some constant $C>0$. By using the upper bound of the following inequality:
\begin{equation}
\label{Ek4}
\frac{1}{\alpha} \leq \coth(\alpha) := \frac{1}{\tanh(\alpha)} \leq \frac{1+\alpha}{\alpha},\quad \alpha>0,
\end{equation}
along with the inequality:
\begin{equation}
\label{shakj}
\frac{1}{\sqrt{2\sinh( \kappa t)}} = \frac{\sqrt{\tanh\left(\frac{\kappa}{2} t\right)}}{\sqrt{2\sinh(\kappa t) \tanh\left(\frac{\kappa}{2} t\right)}} =  \frac{\sqrt{\tanh\left(\frac{\kappa}{2} t\right)}}{2\sinh\left(\frac{\kappa}{2} t\right)} \leq \frac{1}{2\sinh\left(\frac{\kappa}{2} t\right)},
\end{equation}
justified by \eqref{Id1}, then there exists another constant $C>0$ s.t. $\forall L \in [\mathpzc{L}_{\kappa_{0}},\infty)$ and $\forall \kappa \in [\kappa_{0},\infty)$:
\begin{equation}
\label{ZL1}
\forall t>0,\quad \mathpzc{Y}_{L,\kappa}^{(d=1),1}(t) \leq  C \left(1+\sqrt{\kappa}\right)\sqrt{1+\kappa} \frac{(1+t)^{\frac{7}{2}}}{2\sinh\left(\frac{\kappa}{2} t\right)}  \mathrm{e}^{- \frac{\kappa}{8} \frac{L^{2}}{4} \tanh\left(\frac{\kappa}{2}t\right)}.
\end{equation}
Next, we estimate \eqref{calZ2}. From \eqref{R1} and \eqref{Mehler}, one has $\forall L \in [\mathpzc{L}_{\kappa_{0}},\infty)$, $\forall \kappa \in [\kappa_{0},\infty)$ and $\forall t>0$:
\begin{multline*}
\mathpzc{Y}_{L,\kappa}^{(d=1),2}(t) \leq \sqrt{\kappa}\sqrt{1+\kappa} (1+t)^{\frac{7}{2}} \int_{0}^{t} \frac{1}{\sqrt{\sinh(\kappa(t-s))}} \mathrm{e}^{-\frac{\kappa}{16} \frac{L^{2}}{4} \tanh\left(\frac{\kappa}{2}(t-s)\right)}\, \mathrm{d}s\\ \times \int_{\Lambda_{L}^{1}}  G_{\infty,\kappa}^{(d=1)}(x,\varsigma_{L};s,1) G_{\infty,0}^{(d=1)}(\varsigma_{L},x;4(t-s))\, \mathrm{d}x.
\end{multline*}
For the following, we need to make appear from the integration over $\Lambda_{L}^{1}$ a Gaussian decay in $L$ while having the argument $s$. To do so, let us remark that on $\mathbb{R}^{2d}$, $d \in \{1,2,3\}$ one has:
\begin{equation}
\label{astuce}
\forall s>0,\quad G_{\infty,\kappa}^{(d)}(\bold{x},\bold{y};s,1) \leq \mathrm{e}^{- \frac{\kappa}{4}\left(\vert \bold{x}\vert^{2}+ \vert \bold{y}\vert^{2}\right) \tanh\left(\frac{\kappa}{2}s\right)} G_{\infty,\kappa}^{(d)}(\bold{x},\bold{y};s,2).
\end{equation}
To get \eqref{astuce}, we expanded in \eqref{Mehler} the squares and used that $2ab \leq (a^{2}+b^{2})$ with the fact that $\coth(\alpha) - \tanh(\alpha) \geq 0$ $\forall \alpha>0$. Now, from \eqref{astuce} followed by the lower bound:
\begin{multline}
 \label{Id5}
\tanh(\alpha s)+ \tanh(\alpha(t-s)) = \\
\tanh(\alpha t)\left\{1 + \tanh(\alpha s) \tanh(\alpha(t-s))\right\} \geq \tanh(\alpha t),\quad \forall \alpha\geq0, \forall t>s>0,
\end{multline}
then by using the upper bound in the first inequality of \eqref{roughes} along with \eqref{propsemi}, one has $\forall L \in [\mathpzc{L}_{\kappa_{0}},\infty)$ and $\forall \kappa \in [\kappa_{0},\infty)$:
\begin{equation*}
\forall t>0,\quad \mathpzc{Y}_{L,\kappa}^{(d=1),2}(t) \leq C \kappa \sqrt{1+\kappa} \frac{(1+t)^{\frac{7}{2}}}{\sqrt{t}} \mathrm{e}^{-\frac{\kappa}{16} \frac{L^{2}}{4} \tanh\left(\frac{\kappa}{2}t\right)} \int_{0}^{t}  \frac{\sqrt{s}}{\sqrt{\sinh(\kappa s) \sinh(\kappa(t-s))}}\,\mathrm{d}s,
\end{equation*}
for some constant $C>0$. It remains to use successively the identity:
\begin{equation}
\label{Id3}
\coth(\alpha s) + \coth(\alpha(t-s)) = \frac{\sinh(\alpha t)}{\sinh(\alpha s) \sinh(\alpha(t-s))},\quad \forall \alpha>0,\, \forall t>s>0,
\end{equation}
followed by \eqref{Ek4} and \eqref{shakj}. This yields:
\begin{equation}
\label{ZL2}
\forall t>0,\quad \mathpzc{Y}_{L,\kappa}^{(d=1),2}(t) \leq C \sqrt{\kappa} (1+\kappa) \sqrt{t} \frac{(1+t)^{4}}{2 \sinh\left(\frac{\kappa}{2}t\right)} \mathrm{e}^{-\frac{\kappa}{16} \frac{L^{2}}{4} \tanh\left(\frac{\kappa}{2}t\right)},
\end{equation}
for another constant $C>0$. Gathering \eqref{ZL1}-\eqref{ZL2} together, we get \eqref{majYp} in the case of $d=1$.\\
Subsequently, we turn to the case of $d=2$. Since \eqref{diff2} is made up of four terms, then the same holds for the quantity in \eqref{Yp}. Since these terms have the same structure, it is enough to treat only one of them. In view of \eqref{crus}, we need to estimate $\forall L \in [\mathpzc{L}_{\kappa_{0}},\infty)$,  $\forall \kappa \in [\kappa_{0},\infty)$ and $\forall t>0$:
\begin{align}
\label{calZ11}
\mathpzc{Y}_{L,\kappa}^{(d=2),1}(t) &:= \frac{1}{2} \int_{0}^{t}   \int_{\Lambda_{L}^{2}}  \int_{\Lambda_{L}^{1}}  G_{\infty,\kappa}^{(d=2)}(\bold{x},(z_{1},\varsigma_{L});s,1) \mathpzc{P}_{\infty,\kappa}^{(d=2)}((z_{1},\varsigma_{L}),\bold{x};t-s)\, \mathrm{d}z_{1}\,\mathrm{d}\bold{x}\,\mathrm{d}s,\\
\label{calZ21}
\mathpzc{Y}_{L,\kappa}^{(d=2),2}(t) &:= \frac{1}{2} \int_{0}^{t}   \int_{\Lambda_{L}^{2}}  \int_{\Lambda_{L}^{1}}  G_{\infty,\kappa}^{(d=2)}(\bold{x},(z_{1},\varsigma_{L});s,1) \mathpzc{R}_{L,\kappa}^{(d=2)}((z_{1},\varsigma_{L}),\bold{x};t-s)\,\mathrm{d}z_{1}\,\mathrm{d}\bold{x}\,\mathrm{d}s.
\end{align}
The strategy consists in using the property \eqref{multd} in order to use the results stated in the case of $d=1$. Let us first estimate the quantity in \eqref{calZ11}. In view of  \eqref{multd} and \eqref{Q1e}, then from \eqref{propsemi2}:
\begin{multline*}
\int_{\Lambda_{L}^{2}} \int_{\Lambda_{L}^{1}}  G_{\infty,\kappa}^{(d=2)}(\bold{x},(z_{1},\varsigma_{L});s,1) G_{\infty,\kappa}^{(d=2)}((z_{1},\varsigma_{L}),\bold{x};t-s,8)\,\mathrm{d}z_{1}\,\mathrm{d}\bold{x} \\ \leq
C \int_{\mathbb{R}^{1}} G_{\infty,\kappa}^{(d=1)}(x_{1},x_{1};t,8)\, \mathrm{d}x_{1} \int_{\mathbb{R}^{1}}G_{\infty,\kappa}^{(d=1)}(x_{2},\varsigma_{L};s,1) G_{\infty,\kappa}^{(d=1)}(\varsigma_{L},x_{2};t-s,8)\, \mathrm{d}x_{2},
\end{multline*}
for some constant $C>0$. From \eqref{traceGk}, the first integral in the above r.h.s. is nothing but the trace (multiplied by a constant). Then for any $L \in [\mathpzc{L}_{\kappa_{0}},\infty)$, $\kappa \in [\kappa_{0},\infty)$ and $t>0$, we arrive at:
\begin{multline*}
\mathpzc{Y}_{L,\kappa}^{(d=2),1}(t) \leq C \left(1+\sqrt{\kappa}\right)  \frac{(1+t)^{\frac{5}{2}}}{2\sinh\left(\frac{\kappa}{2} t\right)} \\
\times \int_{0}^{t}   \sqrt{\coth\left(\frac{\kappa}{2} (t-s)\right)}\,\mathrm{d}s \int_{\mathbb{R}^{1}}  G_{\infty,\kappa}^{(d=1)}(x_{2},\varsigma_{L};s,1) G_{\infty,\kappa}^{(d=1)}(\varsigma_{L},x_{2};t-s,8)\,\mathrm{d}x_{2},
\end{multline*}
for another $C>0$. The integral w.r.t. $s$ has been estimated in the case of $d=1$, see \eqref{excalY1}. Therefore, it remains to mimic the arguments leading to \eqref{ZL1} to conclude. Next, we estimate \eqref{calZ21}. In view of \eqref{Mehler}-\eqref{multd} and \eqref{R1}, from the first upper bound in \eqref{roughes} along with \eqref{propsemi2}:
\begin{multline*}
\int_{\Lambda_{L}^{2}}  \int_{\Lambda_{L}^{1}} G_{\infty,\kappa}^{(d=2)}(\bold{x},(z_{1},\varsigma_{L});s,1) G_{\infty,0}^{(d=2)}((z_{1},\varsigma_{L}),\bold{x};4(t-s))\,\mathrm{d}z_{1}\,\mathrm{d}\bold{x}  \leq
C \sqrt{\frac{\kappa}{\sinh(\kappa s)}} \sqrt{s} \\
\times \int_{\Lambda_{L}^{1}}  G_{\infty,0}^{(d=1)}(x_{1},x_{1};4t) \,\mathrm{d}x_{1} \int_{\mathbb{R}^{1}} G_{\infty,\kappa}^{(d=1)}(x_{2},\varsigma_{L};s,1) G_{\infty,0}^{(d=1)}(\varsigma_{L},x_{2};4(t-s))\,\mathrm{d}x_{2},
\end{multline*}
for some constant $C>0$. Note that the integrand in the first integral of the above r.h.s. is nothing but a constant. This will make appear a factor $L$, but we will get rid of it at the end. Ergo, in view of \eqref{R1} and \eqref{Mehler}, there exists another $C>0$ s.t. $\forall L \in [\mathpzc{L}_{\kappa_{0}},\infty)$, $\forall \kappa \in [\kappa_{0},\infty)$ and $\forall t>0$:
\begin{multline*}
\mathpzc{Y}_{L,\kappa}^{(d=2),2}(t) \leq C \kappa^{\frac{3}{2}} (1+\kappa) L \frac{(1+t)^{6}}{\sqrt{t}} \int_{0}^{t} \frac{\sqrt{s} \sqrt{t-s}}{\sqrt{\sinh(\kappa s)} \sinh(\kappa(t-s))} \mathrm{e}^{-\frac{\kappa}{16} \frac{L^{2}}{4} \tanh\left(\frac{\kappa}{2}(t-s)\right)} \,\mathrm{d}s \\ \times \int_{\mathbb{R}^{1}} G_{\infty,\kappa}^{(d=1)}(x_{2},\varsigma_{L};s,1) G_{\infty,0}^{(d=1)}(\varsigma_{L},x_{2};4(t-s))\,\mathrm{d}x_{2}.
\end{multline*}
The rest of the proof mimics the strategy used for the case of $d=1$. By using the upper bound in the first inequality of \eqref{roughes} along with \eqref{propsemi}, one has $\forall L \in [\mathpzc{L}_{\kappa_{0}},\infty)$, $\forall \kappa \in [\kappa_{0},\infty)$ and $\forall t>0$:
\begin{equation*}
\mathpzc{Y}_{L,\kappa}^{(d=2),2}(t) \leq C \kappa^{2} (1+\kappa) L \frac{(1+t)^{6}}{t} \int_{0}^{t} \frac{s \sqrt{t-s}}{\sinh(\kappa s)\sinh(\kappa(t-s))} \mathrm{e}^{-\frac{\kappa}{16} \frac{L^{2}}{4} \tanh\left(\frac{\kappa}{2}(t-s)\right)} \mathrm{e}^{-\frac{\kappa}{4} \frac{L^{2}}{4} \tanh\left(\frac{\kappa}{2}s\right)}\,\mathrm{d}s.
\end{equation*}
By using successively \eqref{Id5}, \eqref{Id3} and \eqref{emax} one straightforwardly gets:
\begin{equation*}
\mathpzc{Y}_{L,\kappa}^{(d=2),2}(t) \leq C \kappa (1+\kappa)^{2} L  \sqrt{t} \frac{(1+t)^{7}}{2 \sinh(\kappa t)} \mathrm{e}^{-\frac{\kappa}{16} \frac{L^{2}}{4} \tanh\left(\frac{\kappa}{2}t\right)},
\end{equation*}
for another $L$-independent $C>0$. It remains to use \eqref{redexp} to get rid of the $L$-factor:
\begin{equation*}
\frac{L}{2 \sinh(\kappa t)} \mathrm{e}^{-\frac{\kappa}{32} \frac{L^{2}}{4} \tanh\left(\frac{\kappa}{2}t\right)} \leq   \frac{C}{\sqrt{\kappa}} \frac{1}{2 \sinh(\kappa t) \sqrt{\tanh\left(\frac{\kappa}{2}t\right)}} \leq \frac{C}{\sqrt{\kappa}} \frac{1}{\left(2 \sinh\left(\frac{\kappa}{2} t\right)\right)^{2}}.
\end{equation*}
Gathering the above estimates together, one arrives $\forall L \in [\mathpzc{L}_{\kappa_{0}},\infty)$, $\forall \kappa \in [\kappa_{0},\infty)$ and $\forall t>0$ at:
\begin{equation*}
\mathpzc{Y}_{L,\kappa}^{(d=2),2}(t) \leq C \sqrt{\kappa} (1+\kappa)^{2} \sqrt{t} \frac{(1+t)^{7}}{\left(2 \sinh\left(\frac{\kappa}{2} t\right)\right)^{2}} \mathrm{e}^{-\frac{\kappa}{32} \frac{L^{2}}{4} \tanh\left(\frac{\kappa}{2}t\right)},
\end{equation*}
for another constant $C>0$. The case of $d=3$ can be deduced by similar arguments. \qed

\subsection{Proof of Proposition \ref{cruspro}.}
\label{proofpro}

As it was previously mentioned, Proposition \ref{cruspro} contains the key-estimate to prove Theorem \ref{Mimp}. The proof leans on an approximation of the finite-volume semigroup operator via a geometric perturbation theory. Although this method had been originally developed for the resolvent operators, see \cite{CN1} and also \cite{CFFH,S1}, below we extend the method to the semigroup operators.

\subsubsection{An approximation via a geometric perturbation theory.}
\label{GPTo}

The key-idea consists in isolating in $\Lambda_{L}^{d}$ the region close to the boundary from the bulk where the semigroup $G_{\infty,\kappa}(t)$ will act. The underlying difficulty  is to keep a good control of the remainder terms arising from this approximation. This will be achieved by using well-chosen cutoff functions.\\
\indent For any $0 < \eta <1$, $0 < \vartheta \leq 1000$, $d \in \{1,2,3\}$ and $L \in (0,\infty)$ define:
\begin{equation}
\label{Theta}
\Theta_{L,\eta}(\vartheta) := \left\{\bold{x} \in \overline{\Lambda_{L}^{d}}: \mathrm{dist}\left(\bold{x}, \partial \Lambda_{L}^{d}\right) \leq \vartheta L^{\eta}\right\}.
\end{equation}
For $L$ sufficiently large, $\Theta_{L,\eta}(\vartheta)$ models a 'thin' compact subset of $\Lambda_{L}^{d}$ near the boundary with Lebesgue-measure $\vert \Theta_{L,\eta}(\vartheta)\vert$ of order $\mathcal{O}(L^{(d-1)+\eta})$.
For any $0 < \eta < 1$, let $L_{0}=L_{0}(\eta) \geq 1$ s.t.
\begin{equation}
\label{defR0}
\Theta_{L_{0},\eta}(1000) \subsetneq \Lambda_{L_{0}}^{d},\quad L_{0} - L_{0}^{\eta} \geq L_{0}/\sqrt{2},
\end{equation}
and $L_{0}$ large enough. Let us now introduce some well-chosen families of smooth cutoff functions. \\
Let $f_{L,\eta}$ and $f_{L,\eta}^{c}$, $L \in [L_{0}(\eta),\infty)$ be a partition of the unity of $\Lambda_{L}^{d}$ satisfying:
\begin{gather*}
f_{L,\eta} + f_{L,\eta}^{c} = 1\quad \textrm{on $\Lambda_{L}^{d}$};\\
\mathrm{Supp}\left(f_{L,\eta}\right) \subset \left(\Lambda_{L}^{d} \setminus \Theta_{L,\eta}\left(\frac{1}{16}\right)\right),\quad
f_{L,\eta} = 1\,\, \textrm{if}\,\, \bold{x} \in \left(\Lambda_{L}^{d} \setminus \Theta_{L,\eta}\left(\frac{1}{8}\right)\right),\quad 0 \leq f_{L,\eta} \leq 1;\\
\mathrm{Supp}\left(f_{L,\eta}^{c}\right) \subset \Theta_{L,\eta}\left(\frac{1}{8}\right), \quad
f_{L,\eta}^{c} = 1\,\, \textrm{if}\,\, \bold{x} \in \Theta_{L,\eta}\left(\frac{1}{16}\right).
\end{gather*}
Moreover, there exists a constant $C>0$ s.t.
\begin{equation*}
\forall L\geq L_{0}(\eta),\quad  \left\Vert D^{\sigma} f_{L,\eta} \right\Vert_{\infty}\leq C  L^{- \vert \sigma \vert \eta},\quad \forall \vert \sigma \vert \leq 2,\,\vert \sigma \vert = \sigma_{1}+\dotsb + \sigma_{d}.
\end{equation*}
Also, let $\hat{f}_{L,\eta}$ and $\hat{\hat{f}}_{L,\eta}$, $L \in [L_{0}(\eta),\infty)$ satisfying:
\begin{gather*}
\mathrm{Supp}\left(\hat{f}_{L,\eta}\right) \subset \left(\Lambda_{L}^{d} \setminus \Theta_{L,\eta}\left(\frac{1}{64}\right)\right),\quad
\hat{f}_{L,\eta} = 1\,\, \textrm{if}\,\, \bold{x} \in \left(\Lambda_{L}^{d} \setminus \Theta_{L,\eta}\left(\frac{1}{32}\right)\right),\quad 0 \leq \hat{f}_{L,\eta} \leq 1;\\
\mathrm{Supp}\left(\hat{\hat{f}}_{L,\eta}\right) \subset \Theta_{L,\eta}\left(\frac{1}{2}\right),\quad
\hat{\hat{f}}_{L,\eta} = 1\,\, \textrm{if}\,\, \bold{x} \in \Theta_{L,\eta}\left(\frac{1}{4}\right),\quad 0 \leq \hat{\hat{f}}_{L,\eta} \leq 1.
\end{gather*}
Moreover, there exists another constant $C>0$ s.t.
\begin{equation*}
\forall L \geq L_{0}(\eta),\quad \max\left\{\left\Vert D^{\sigma} \hat{f}_{L,\eta} \right\Vert_{\infty}, \left\Vert D^{\sigma} \hat{\hat{f}}_{L,\eta} \right\Vert_{\infty}\right\} \leq C  L^{- \vert \sigma \vert \eta}, \quad \forall \vert \sigma \vert \leq 2.
\end{equation*}
With these properties, one straightforwardly gets:
\begin{gather}
\label{supot1}
\hat{f}_{L,\eta} f_{L,\eta} = f_{L,\eta};\\
\label{disjsup1}
\mathrm{dist}\left(\mathrm{Supp}\left(D^{\sigma} \hat{f}_{L,\eta}\right), \mathrm{Supp}\left( D^{\tau} f_{L,\eta}\right)\right) \geq C L^{\eta}, \quad \forall 1\leq \vert \sigma\vert \leq 2,\,\forall 0 \leq \vert \tau\vert \leq 2;\\
\label{supot2}
\hat{\hat{f}}_{L,\eta}f_{L,\eta}^{c} = f_{L,\eta}^{c};\\
\label{disjsup2}
\mathrm{dist}\left(\mathrm{Supp}\left(D^{\sigma} \hat{\hat{f}}_{L,\eta}\right), \mathrm{Supp}\left(D^{\tau} f_{L,\eta}^{c}\right)\right) \geq C L^{\eta}, \quad \forall 1\leq \vert \sigma \vert \leq 2,\, \forall 0 \leq \vert \tau\vert \leq 2,
\end{gather}
for some $L$-independent constant $C>0$.\\
\indent Afterwards, let us define $\forall 0<\eta<1$, $\forall L \in [L_{0}(\eta),\infty)$ (see \eqref{defR0}) and $\forall \kappa>0$ on $\mathcal{C}_{0}^{\infty}(\Lambda_{L}^{d})$:
\begin{equation}
\label{ptih}
h_{L,\kappa,\eta} := \frac{1}{2} \left(-i\nabla_{\bold{x}}\right)^{2} + \frac{1}{2} \kappa^{2} V_{L,\eta}(\bold{x}),\quad V_{L,\eta}(\bold{x}) := \left\{\begin{array}{ll}
\vert \bold{x}\vert^{2},\,\,\,&\textrm{if $\bold{x} \in \mathrm{Supp}\left(\hat{\hat{f}}_{L,\eta}\right)$},\\
\frac{1}{4}\left(L-L^{\eta}\right)^{2},\,\,\,&\textrm{otherwise}.
\end{array}\right.
\end{equation}
By standard arguments, \eqref{ptih} extends to a family of self-adjoint and semi-bounded operators $\forall L \in [L_{0}(\eta),\infty)$, denoted again by $h_{L,\kappa,\eta}$, with domain $D(h_{L,\kappa,\eta})=W_{0}^{1,2}(\Lambda_{L}^{d})\cap W^{2,2}(\Lambda_{L}^{d})$.
For any $0<\eta<1$, $L \in [L_{0}(\eta),\infty)$ and $\kappa>0$, let $\{g_{L,\kappa,\eta}(t):= \mathrm{e}^{- t h_{L,\kappa,\eta}}:L^{2}(\Lambda_{L}^{d}) \rightarrow L^{2}(\Lambda_{L}^{d})\}_{t \geq 0}$ be the strongly-continuous one-parameter semigroup generated by $h_{L,\kappa,\eta}$. It is an integral operator with an integral kernel jointly continuous in $(\bold{x},\bold{y},t) \in \overline{\Lambda_{L}^{d}}\times \overline{\Lambda_{L}^{d}}\times (0,\infty)$. We denote it by $g_{L,\kappa,\eta}^{(d)}$.\\
\indent Next, introduce $\forall 0<\eta<1$, $\forall L \in [L_{0}(\eta),\infty)$ and $\forall  \kappa>0$ the following operators on $L^{2}(\Lambda_{L}^{d})$:
\begin{gather}
\label{calG}
\forall t>0,\quad \mathcal{G}_{L,\kappa,\eta}(t) := \hat{f}_{L,\eta} G_{\infty,\kappa}(t) f_{L,\eta} + \hat{\hat{f}}_{L,\eta} g_{L,\kappa,\eta}(t) f_{L,\eta}^{c},\\
\label{calR}
\begin{split}
\mathcal{W}_{L,\kappa,\eta}(t) := &-\left\{\frac{1}{2}\left(\Delta \hat{f}_{L,\eta}\right) + i\left(\nabla \hat{f}_{L,\eta}\right) \cdot \left(-i\nabla\right) \right\} G_{\infty,\kappa}(t) f_{L,\eta} + \\
&- \left\{\frac{1}{2}\left(\Delta \hat{\hat{f}}_{L,\eta}\right) + i\left(\nabla \hat{\hat{f}}_{L,\eta}\right) \cdot \left(-i\nabla\right) \right\}g_{L,\kappa,\eta}(t) f_{L,\eta}^{c}.
\end{split}
\end{gather}
Sometimes, we will use the shorthand notations:
\begin{equation}
\label{splicalG}
\forall t>0,\quad \mathcal{G}_{L,\kappa,\eta}^{(p)}(t) := \hat{f}_{L,\eta}G_{\infty,\kappa}(t) f_{L,\eta},\quad \mathcal{G}_{L,\kappa,\eta}^{(r)}(t) := \hat{\hat{f}}_{L,\eta}g_{L,\kappa,\eta}(t) f_{L,\eta}^{c}.
\end{equation}

The main result of this paragraph is the following Duhamel-like formula:

\begin{proposition}
\label{Duhamcom}
$\forall d \in \{1,2,3\}$, $\forall 0<\eta<1$, $\forall L\in [L_{0}(\eta),\infty)$ and  $\forall \kappa>0$, it takes place in the bounded operators sense on $L^{2}(\Lambda_{L}^{d})$:
\begin{equation}
\label{Duham}
\forall t>0,\quad G_{L,\kappa}(t) = \mathcal{G}_{L,\kappa,\eta}(t) - \int_{0}^{t}  G_{L,\kappa}(t-s) \mathcal{W}_{L,\kappa,\eta}(s)\,\mathrm{d}s.
\end{equation}
\end{proposition}

The proof of Proposition \ref{Duhamcom} can be found in Sec. \ref{Duha}; it is essentially based on the application of \cite[Prop. 3]{C1} taking into account the features of the cutoff functions introduced previously.

\begin{remark}
\label{allbon}
One can derive the following upper bounds on the operator norms. $\forall d \in \{1,2,3\}$ there exist two constants $C_{d},c>0$ s.t. $\forall 0<\eta < 1$, $\forall L\in [L_{0}(\eta),\infty)$, $\forall \kappa>0$ and $\forall t>0$:
\begin{gather}
\label{normcalG}
\left\Vert \mathcal{G}_{L,\kappa,\eta}(t)\right\Vert \leq \left\Vert \mathcal{G}_{L,\kappa,\eta}^{(p)}(t)\right\Vert + \left\Vert \mathcal{G}_{L,\kappa,\eta}^{(r)}(t)\right\Vert \leq \left(\cosh(\kappa t)\right)^{-\frac{d}{2}} + C_{d} \mathrm{e}^{-\frac{\kappa^{2}}{16}L^{2}t},\\
\label{normcalR}
\left\Vert \mathcal{W}_{L,\kappa,\eta}(t)\right\Vert \leq C_{d}\sqrt{1+\kappa} \frac{\sqrt{1+t}}{\sqrt{t}} \mathrm{e}^{-c \frac{L^{2\eta}}{t}} \left\{1 + (1+t)^{d-\frac{1}{2}} \mathrm{e}^{-\frac{\kappa^{2}}{8} \frac{L^{2}}{4} t}\right\}.
\end{gather}
The upper bound in \eqref{normcalG} comes from \eqref{norm} and \eqref{kerhL}. The rough estimate in \eqref{normcalR} is derived from Lemmas \ref{lem7} and \ref{plentes} along with the properties \eqref{disjsup1}-\eqref{disjsup2}.
\end{remark}


\subsubsection{Conclusion of the proof.}

The starting-point in the proof of Proposition \ref{cruspro} is the Duhamel-like formula in \eqref{Duham}. Taking its adjoint, one has $\forall d\in \{1,2,3\}$, $\forall 0<\eta<1$, $\forall L \in [L_{0}(\eta),\infty)$ (see \eqref{defR0}) and $\forall \kappa>0$ on $L^{2}(\Lambda_{L}^{d})$:
\begin{equation}
\label{adjDu}
\forall t>0,\quad G_{L,\kappa}(t) = \mathcal{G}_{L,\kappa,\eta}^{*}(t) - \int_{0}^{t} \mathcal{W}^{*}_{L,\kappa,\eta}(s) G_{L,\kappa}(t-s)\, \mathrm{d}s,
\end{equation}
where the adjoint operator of $\mathcal{G}_{L,\kappa,\eta}(t)$ and $\mathcal{W}_{L,\kappa,\eta}(t)$ reads respectively as, see \eqref{calG}-\eqref{calR}:
\begin{gather}
\label{calGad}
\mathcal{G}^{*}_{L,\kappa,\eta}(t) = f_{L,\eta}G_{\infty,\kappa}(t) \hat{f}_{L,\eta} + f_{L,\eta}^{c}g_{L,\kappa,\eta}(t) \hat{\hat{f}}_{L,\eta},\\
\label{calRad}
\begin{split}
&\mathcal{W}^{*}_{L,\kappa,\eta}(t) = - f_{L,\eta} G_{\infty,\kappa}(t) \frac{1}{2}\left(\Delta\hat{f}_{L,\eta}\right) + i f_{L,\eta}\left\{\left(-i \nabla\right) G_{\infty,\kappa}(t) - \left[\left(-i\nabla\right),G_{\infty,\kappa}(t)\right]\right\} \left(\nabla\hat{f}_{L,\eta}\right) +\\
&- f_{L,\eta}^{c} g_{L,\kappa,\eta}(t) \frac{1}{2}\left(\Delta \hat{\hat{f}}_{L,\eta}\right) + i f_{L,\eta}^{c}\left\{\left(-i \nabla\right) g_{L,\kappa,\eta}(t) - \left[\left(-i\nabla\right),g_{L,\kappa,\eta}(t)\right]\right\}\left(\nabla \hat{\hat{f}}_{L,\eta}\right).
\end{split}
\end{gather}
Here, $[\cdot\,,\cdot\,]$ denotes the usual commutator, and in the bounded operators sense:
\begin{gather}
\label{commut1}
\left[\left(-i\nabla\right),G_{\infty,\kappa}(t)\right] = - \int_{0}^{t}  G_{\infty,\kappa}(t-s)\left[\left(-i\nabla\right), H_{\infty,\kappa}\right] G_{\infty,\kappa}(s)\,\mathrm{d}s,\\
\label{commut2}
\left[\left(-i\nabla\right),g_{L,\kappa,\eta}(t)\right] = - \int_{0}^{t}  g_{L,\kappa,\eta}(t-s)\left[\left(-i\nabla\right), h_{L,\kappa,\eta}\right] g_{L,\kappa,\eta}(s)\,\mathrm{d}s.
\end{gather}
Writing \eqref{adjDu} in the kernels sense, it follows the identity:
\begin{multline}
\label{adjDuke}
\forall (\bold{x},\bold{y})\in \Lambda_{L}^{2d},\,\forall t>0,\quad \nabla_{\bold{x}}G_{L,\kappa}^{(d)}(\bold{x},\bold{y};t) = \\  \nabla_{\bold{x}}\left(\mathcal{G}_{L,\kappa,\eta}^{*}\right)^{(d)}(\bold{x},\bold{y};t)  - \int_{0}^{t}  \int_{\Lambda_{L}^{d}} \nabla_{\bold{x}}\left(\mathcal{W}^{*}_{L,\kappa,\eta}\right)^{(d)}(\bold{x},\bold{z};s) G_{L,\kappa}^{(d)}(\bold{z},\bold{y};t-s)\,\mathrm{d}\bold{z}\,\mathrm{d}s.
\end{multline}

Next, we need the following lemma whose proof can be found in Sec. \ref{intermB}:

\begin{lema}
\label{toues}
$\forall d \in \{1,2,3\}$ there exist two constants $c,C_{d}>0$ s.t.:\\
$\mathrm{(i)}$ $\forall 0<\eta<1$, $\forall L \in [L_{0}(\eta),\infty)$, $\forall \kappa>0$, $\forall(\bold{x},\bold{y}) \in \Lambda_{L}^{2d}$ and $\forall t>0$:
\begin{gather}
\left\vert \nabla_{\bold{x}} \left(\mathcal{G}^{*}_{L,\kappa,\eta}\right)^{(d)}(\bold{x},\bold{y};t)\right\vert \leq C_{d} \left\{P_{\infty,\kappa,\eta}^{(d)}(\bold{x},\bold{y};t) + R_{L,\kappa,\eta}^{(d)}(\bold{x},\bold{y};t)\right\},\nonumber\\
\label{kerud}
P_{\infty,\kappa,\eta}^{(d)}(\bold{x},\bold{y};t) :=  \left(1+\sqrt{\kappa}\right) \sqrt{\coth\left(\frac{\kappa}{2} t\right)} G_{\infty,\kappa}^{(d)}(\bold{x},\bold{y};t,2),\\
\label{kervd}
R_{L,\kappa,\eta}^{(d)}(\bold{x},\bold{y};t) := \frac{(1+t)^{d}}{\sqrt{t}} \mathrm{e}^{-\frac{\kappa^{2}}{8} \frac{L^{2}}{4} t} G_{\infty,0}^{(d)}(\bold{x},\bold{y};2t).
\end{gather}
$\mathrm{(ii)}$. $\forall \frac{1}{4}<\eta<1$, $\forall L \in [L_{0}(\eta),\infty)$, $\forall \kappa>0$, $\forall(\bold{x},\bold{y}) \in \Lambda_{L}^{2d}$ and $\forall t>0$:
\begin{gather}
\label{kcalRstar}
\left\vert \nabla_{\bold{x}} \left(\mathcal{W}^{*}_{L,\kappa,\eta}\right)^{(d)}(\bold{x},\bold{y};t)\right\vert \leq C_{d}\left\{ r_{\infty,\kappa,\eta}^{(d)}(\bold{x},\bold{y};t) + r_{L,\kappa,\eta}^{(d)}(\bold{x},\bold{y};t)\right\},\\
\label{kerwd1}
r_{\infty,\kappa,\eta}^{(d)}(\bold{x},\bold{y};t) := \left(1+\sqrt{\kappa}\right) \sqrt{\coth\left(\frac{\kappa}{2} t\right)} (1+t) \mathrm{e}^{-c \kappa L^{2\eta} \coth\left(\frac{\kappa}{2}t\right)} G_{\infty,\kappa}^{(d)}(\bold{x},\bold{y};t,8),\\
\label{kerwd2}
r_{L,\kappa,\eta}^{(d)}(\bold{x},\bold{y};t) := \frac{(1+t)^{2d}}{\sqrt{t}} \mathrm{e}^{-\frac{\kappa^{2}}{16} \frac{L^{2}}{4} t} \mathrm{e}^{-c \frac{L^{2\eta}}{t}} \chi_{\Theta_{L,\eta}(\frac{1}{8})}(\bold{x}) G_{\infty,0}^{(d)}(\bold{x},\bold{y};4t) \chi_{\Theta_{L,\eta}(\frac{1}{2})}(\bold{y}).
\end{gather}
Here, $\chi_{\Theta_{L,\eta}(\vartheta)}$, $\vartheta>0$ denotes the indicator function associated with $\Theta_{L,\eta}(\vartheta)$ defined in \eqref{Theta}.
\end{lema}

\begin{remark} In $\mathrm{(ii)}$, $\eta$ has been restricted to $(\frac{1}{4},1)$ only to make the estimates more elegant.
\end{remark}

\noindent \textit{Proof of Proposition \ref{cruspro}.} From now on, we set $\eta=\frac{1}{2}$ in the r.h.s. of \eqref{adjDuke}. In view of the second term, \eqref{kcalRstar} with \eqref{kerwd1}-\eqref{kerwd2} and \eqref{fondineq}, we need to estimate the two quantities:
\begin{gather}
\label{calQ1}
\mathpzc{Q}_{\infty,\kappa}^{(d)}(\bold{x},\bold{y};t) := \int_{0}^{t}  \int_{\mathbb{R}^{d}} r_{\infty,\kappa,\eta=\frac{1}{2}}^{(d)}(\bold{x},\bold{z};s) G_{\infty,\kappa}^{(d)}(\bold{z},\bold{y};t-s,1)\,\mathrm{d}\bold{z}\,\mathrm{d}s,\\
\label{calQ2}
\mathpzc{Q}_{L,\kappa}^{(d)} (\bold{x},\bold{y};t) := \int_{0}^{t}  \int_{\mathbb{R}^{d}} r_{L,\kappa,\eta=\frac{1}{2}}^{(d)}(\bold{x},\bold{z};s) G_{\infty,\kappa}^{(d)}(\bold{z},\bold{y};t-s,1)\,\mathrm{d}\bold{z}\,\mathrm{d}s.
\end{gather}
Let $L \geq L_{0}(\eta=\frac{1}{2})$ defined in \eqref{defR0}. We start with \eqref{calQ1}. From \eqref{kerwd1} followed by \eqref{propsemi2}, there exist two constants $c,C>0$ s.t. $\forall \kappa>0$, $\forall (\bold{x},\bold{y}) \in \Lambda_{L}^{2d}$ and $\forall t>0$:
\begin{equation*}
\mathpzc{Q}_{\infty,\kappa}^{(d)}(\bold{x},\bold{y};t) \leq C \left(1+\sqrt{\kappa}\right) (1+t) \frac{\sqrt{\coth\left(\frac{\kappa}{2} t\right)}}{\sqrt{\coth\left(\frac{\kappa}{2} t\right)}} G_{\infty,\kappa}^{(d)}(\bold{x},\bold{y};t,8) \int_{0}^{t} \sqrt{\coth\left(\frac{\kappa}{2} s\right)} \mathrm{e}^{-c \kappa L \coth\left(\frac{\kappa}{2} s\right)}\, \mathrm{d}s.
\end{equation*}
By using \eqref{redexp} to get rid of the $\coth$ in the integrand, followed by the lower bound in \eqref{Ek4} for the (artificial) denominator in the  above r.h.s., then the upper bound in \eqref{Q1e} follows. Let us turn to the quantity in \eqref{calQ2}. From \eqref{kerwd2}, one has $\forall \kappa>0$, $\forall (\bold{x},\bold{y}) \in \Lambda_{L}^{2d}$ and $\forall t>0$:
\begin{equation*}
\mathpzc{Q}_{L,\kappa}^{(d)}(\bold{x},\bold{y};t) \leq  (1+t)^{2d}
\int_{0}^{t} \frac{ \mathrm{e}^{-\frac{\kappa^{2}}{16} \frac{L^{2}}{4} s}}{\sqrt{s}}\, \mathrm{d}s \int_{\mathbb{R}^{d}}   G_{\infty,0}^{(d)}(\bold{x},\bold{z};4s) \chi_{\Theta_{L,\frac{1}{2}}(\frac{1}{2})}(\bold{z}) G_{\infty,\kappa}^{(d)}(\bold{z},\bold{y};t-s,1)\,\mathrm{d}\bold{z}.
\end{equation*}
To make appear a Gaussian decay in $L$ from the integration over $\mathbb{R}^{d}$, we use \eqref{astuce}. Here, the presence of the characteristic function in the integrand plays a crucial role. Since $\forall \bold{z} \in \Theta_{L,\frac{1}{2}}(\frac{1}{2})$, $\vert \bold{z}\vert \geq \frac{1}{2}(L-\sqrt{L})$ (recall that $L\geq L_{0}(\eta=\frac{1}{2})$, $L_{0}(\eta)$ as in \eqref{defR0}) then  $\forall (\bold{x},\bold{y}) \in \Lambda_{L}^{2d}$ and $\forall t>0$:
\begin{multline}
\label{multsepa}
\mathpzc{Q}_{L,\kappa}^{(d)}(\bold{x},\bold{y};t) \leq  (1+t)^{2d}
\int_{0}^{t} \frac{1}{\sqrt{s}} \mathrm{e}^{-\frac{\kappa^{2}}{16} \frac{L^{2}}{4} s}  \mathrm{e}^{-\frac{\kappa}{8}\frac{L^{2}}{4} \tanh\left(\frac{\kappa}{2}(t-s)\right)}\,\mathrm{d}s \\
\times \int_{\Theta_{L,\eta}(\frac{1}{2})} G_{\infty,0}^{(d)}(\bold{x},\bold{z};4s) G_{\infty,\kappa}^{(d)}(\bold{z},\bold{y};t-s,2)\,\mathrm{d}\bold{z}.
\end{multline}
Next, we extend the integration w.r.t. $\bold{z}$ to $\mathbb{R}^{d}$ and we use the first upper bound in \eqref{roughes} followed by \eqref{propsemi}. Then, we introduce a factor $s^{\frac{d-1}{2}} s^{-\frac{d-1}{2}}$ under the integral w.r.t $s$. Next, we successively use the lower bound in \eqref{Ek4} and the upper bound $\cosh(\alpha)\leq \mathrm{e}^{\alpha}$ $\forall \alpha \geq 0$ leading to $(\kappa s)^{-\frac{d}{2}} \leq (\coth(\kappa s))^{\frac{d}{2}} \leq \mathrm{e}^{\frac{d}{2} \kappa s} (\sinh(\kappa s))^{-\frac{d}{2}}$. On this way, we get under the same  conditions than \eqref{multsepa}:
\begin{equation*}
\mathpzc{Q}_{L,\kappa}^{(d)}(\bold{x},\bold{y};t) \leq  C \kappa^{d} (1+t)^{2d} G_{\infty,0}^{(d)}(\bold{x},\bold{y};4t)
\int_{0}^{t} \frac{ s^{\frac{d-1}{2}} (t-s)^{\frac{d}{2}}\, \mathrm{e}^{\frac{d}{2} \kappa s} \mathrm{e}^{-\frac{\kappa^{2}}{16}\frac{L^{2}}{4} s}}{\left\{\sinh(\kappa s) \sinh(\kappa (t-s))\right\}^{\frac{d}{2}}}  \mathrm{e}^{-\frac{\kappa}{8}\frac{L^{2}}{4} \tanh\left(\frac{\kappa}{2}(t-s)\right)}\,\mathrm{d}s,
\end{equation*}
for another constant $C>0$. Here, we artificially made appear a $(\sinh(\kappa s))^{\frac{d}{2}}$ under the integration w.r.t. $s$. This leads to the appearance of a $\mathrm{e}^{\frac{d}{2} \kappa s}$ in the numerator. If $\kappa \geq 1$, we can get rid of it via the term $\mathrm{e}^{-\frac{\kappa^{2}}{16}\frac{L^{2}}{4} s}$ (for $L$ large enough) since $\kappa \leq \kappa^{2}$. If $0<\kappa<1$, one has to choose $L$ large enough accordingly to $\kappa$, i.e. $L \geq cste/ \sqrt{\kappa}$. Given a $\kappa_{0}>0$, let $\mathpzc{L}=\mathpzc{L}_{\kappa_{0}}\geq L_{0}(\frac{1}{2})$ s.t. $\forall L \geq \mathpzc{L}_{\kappa_{0}}$, the inequality $\mathrm{e}^{- \frac{\kappa_{0}^{2}}{2}(\frac{1}{8} \frac{L^{2}}{4} - \frac{d}{\kappa_{0}})s} \leq \mathrm{e}^{- \frac{\kappa_{0}^{2}}{32}\frac{L^{2}}{4}s}$ holds. By using an inequality of type:
\begin{equation}
\label{inegtanh}
\mathrm{e}^{-\frac{\kappa^{2}}{16} \frac{L^{2}}{4} s} \mathrm{e}^{-\frac{\kappa}{8} \frac{L^{2}}{4} \tanh\left(\frac{\kappa}{2}(t-s)\right)}  \leq \mathrm{e}^{-\frac{\kappa}{8} \frac{L^{2}}{4} \left[\tanh\left(\frac{\kappa}{2}s\right)+\tanh\left(\frac{\kappa}{2}(t-s)\right)\right]}
\leq \mathrm{e}^{-\frac{\kappa}{8} \frac{L^{2}}{4} \tanh\left(\frac{\kappa}{2}t\right)},\,\,\, 0<s<t,
\end{equation}
which is justified by the lower bound in \eqref{Ek4} together with \eqref{Id5}, it follows from the identity in \eqref{Id3} that $\forall L \in [\mathpzc{L}_{\kappa_{0}},\infty)$, $\forall \kappa \in [\kappa_{0},\infty)$, $\forall (\bold{x},\bold{y}) \in \Lambda_{L}^{2d}$ and $\forall t>0$:
\begin{multline*}
\mathpzc{Q}_{L,\kappa}^{(d)}(\bold{x},\bold{y};t) \leq  C \kappa^{d} \frac{(1+t)^{2d}}{(\sinh(\kappa t))^{\frac{d}{2}}} \mathrm{e}^{-\frac{\kappa}{16}\frac{L^{2}}{4} \tanh\left(\frac{\kappa}{2} t\right)} G_{\infty,0}^{(d)}(\bold{x},\bold{y};4t) \times \\
\times  \int_{0}^{t} s^{\frac{d-1}{2}} (t-s)^{\frac{d}{2}}\,\left\{(\coth(\kappa s))^{\frac{d}{2}} + (\coth(\kappa (t-s)))^{\frac{d}{2}}\right\}\,\mathrm{d}s,
\end{multline*}
for some constant $C>0$. Here, we used that $(a+b)^{\delta} \leq 2^{\delta}(a^{\delta}+b^{\delta})$ $\forall a,b,\delta>0$. To conclude this estimate, it remains to use that there exists another constant $C>0$ s.t. $\forall t>0$ and $\forall \kappa>0$:
\begin{multline}
\label{emax}
\max\left\{\int_{0}^{t} s^{\frac{d-1}{2}} (t-s)^{\frac{d}{2}} \left(\coth(\kappa s)\right)^{\frac{d}{2}}\,\mathrm{d}s, \int_{0}^{t} s^{\frac{d-1}{2}} (t-s)^{\frac{d}{2}} \left(\coth(\kappa (t-s))\right)^{\frac{d}{2}}\,\mathrm{d}s\right\} \\
\leq C \frac{(1+\kappa)^{\frac{d}{2}}}{\kappa^{\frac{d}{2}}} (1+t)^{\frac{d}{2}} t^{\frac{d+1}{2}}.
\end{multline}
To derive \eqref{R1}, we have to modify the upper bound in \eqref{kervd} by mimicking the method used above to make appear the singularity $(\sinh(\kappa t))^{\frac{d}{2}}$ in the denominator (instead of $\sqrt{t}$). \qed

\subsubsection{Proof of Proposition \ref{Duhamcom}.}
\label{Duha}

The proof relies on \cite[Prop. 3]{C1} that we reproduce here for reader's convenience:

\begin{proposition}
\label{resCorn}
Let $\mathscr{H}$ be a separable Hilbert space and $H$ be a self-adjoint and positive operator having the domain $D \subset \mathscr{H}$. Fix $t_{0}>0$. Assume that there exists an application $(0,t_{0}] \owns t \mapsto S(t) \in \mathfrak{B}(\mathscr{H})$ (the algebra of bounded operators on $\mathscr{H}$) with the following properties:\\
$\mathrm{(A)}$. $\sup_{0<t\leq t_{0}} \Vert S(t)\Vert \leq c_{1} < \infty$. $\mathrm{(B)}$. It is strongly differentiable, $\mathrm{Ran}(S(t)) \subset D$ and $s-\lim_{t \downarrow 0} S(t) = \mathbbm{1}$.
$\mathrm{(C)}$. There exists an application $(0,t_{0}] \owns t \mapsto R(t) \in \mathfrak{B}(\mathscr{H})$ continuous in the operator-norm sense s.t. $\Vert R(t)\Vert \leq c_{2}t^{-\alpha}$ where $0\leq\alpha<1$, and:
\begin{equation*}
\frac{\partial S}{\partial t}(t) \phi + H S(t) \phi = R(t)\phi.
\end{equation*}
Then the following two statements are true:\\
$\mathrm{(i)}$. The sequence of bounded operators $(n > [1/t])$:
\begin{equation*}
T_{n}(t) := \int_{\frac{1}{n}}^{t - \frac{1}{n}}  \exp[-(t-s)H]R(s)\,\mathrm{d}s,
\end{equation*}
converges in norm; let $T(t)$ be its limit;\\
$\mathrm{(ii)}$. The following equality takes place on $\mathfrak{B}(\mathscr{H})$: $\exp(-t H) = S(t) - T(t)$.
\end{proposition}

Before giving the proof, we need a series of estimates related to the kernel of the semigroup generated by the operator in \eqref{ptih}. The proof of the below lemma can be found in Sec. \ref{intermB}.

\begin{lema}
\label{lem7}
$\forall d \in \{1,2,3\}$ there exists a constant $C_{d}>0$ s.t. $\forall 0<\eta<1$, $\forall L\in [L_{0}(\eta),\infty)$, $\forall \kappa>0$, $\forall(\bold{x},\bold{y})\in \Lambda_{L}^{2d}$ and $\forall t>0$:
\begin{gather}
\label{kerhL}
g_{L,\kappa,\eta}^{(d)} (\bold{x},\bold{y};t) \leq C_{d} \mathrm{e}^{- \frac{\kappa^{2}}{4} \frac{L^{2}}{4} t} G_{\infty,0}^{(d)}(\bold{x},\bold{y};t),\\
\label{majgradh}
\left\vert \nabla_{\bold{x}} g_{L,\kappa,\eta}^{(d)}(\bold{x},\bold{y};t)\right\vert
\leq C_{d} \frac{(1+t)^{d}}{\sqrt{t}} \mathrm{e}^{-\frac{\kappa^{2}}{8} \frac{L^{2}}{4} t} G_{\infty,0}^{(d)}(\bold{x},\bold{y};2t),\\
\label{majlaph0}
\left\vert \Delta_{\bold{x}} g_{L,\kappa,\eta}^{(d)} (\bold{x},\bold{y};t)\right\vert
\leq C_{d} \frac{(1+t)^{2d}}{t} \mathrm{e}^{-\frac{\kappa^{2}}{16} \frac{L^{2}}{4} t} G_{\infty,0}^{(d)}(\bold{x},\bold{y};2t).
\end{gather}
\end{lema}

\noindent \textit{Proof of Proposition \ref{Duhamcom}.} The only thing we have to do is verify the assumptions of Proposition \ref{resCorn} in which $\mathcal{G}_{L,\kappa,\eta}(t)$ plays the role of $S(t)$.
Let $0<\eta<1$, $L \in [L_{0}(\eta),\infty)$ and $\kappa>0$ kept fixed. (A) From \eqref{normcalG}, $\mathcal{G}_{L,\kappa,\eta}(t)$ is uniformly bounded in $t$ by some constant $C_{d}>0$.
(B) By using that $s-\lim_{t \downarrow 0} G_{\infty,\kappa}(t)= \mathbbm{1}$ and $s-\lim_{t \downarrow 0} g_{L,\kappa,\eta}(t)= \mathbbm{1}$ in the kernels sense, then:
\begin{equation*}
\forall \phi \in L^{2}(\Lambda_{L}^{d}),\quad \lim_{t \downarrow 0} \mathcal{G}_{L,\kappa,\eta}\phi = \left\{\hat{f}_{L,\eta} f_{L,\eta} + \hat{\hat{f}}_{L,\eta}f_{L,\eta}^{c}\right\}\phi = \left\{f_{L,\eta} + f_{L,\eta}^{c}\right\}\phi=\phi,
\end{equation*}
where we used \eqref{supot1} and \eqref{supot2}. Next, let us investigate the strong differentiability. From \eqref{splicalG}:
\begin{multline*}
\forall \phi \in L^{2}(\Lambda_{L}^{d}),\quad \frac{1}{\delta t}\left\{\left(\mathcal{G}_{L,\kappa,\eta}^{(p)}(t+\delta t) \phi\right)(\cdot\,) - \left(\mathcal{G}_{L,\kappa,\eta}^{(p)}(t)\phi\right)(\cdot\,)\right\} \\
= \hat{f}_{L,\eta}(\cdot\,) \frac{1}{\delta t}  \int_{\mathbb{R}^{d}}  \left\{\int_{\mathbb{R}^{d}}  G_{\infty,\kappa}^{(d)}(\cdot\,,\bold{z};t) G_{\infty,\kappa}^{(d)}(\bold{z},\bold{y};\delta t)\,\mathrm{d}\bold{z} -  G_{\infty,\kappa}^{(d)}(\cdot\,,\bold{y};t)\right\} f_{L,\eta}(\bold{y}) \phi(\bold{y})\,\mathrm{d}\bold{y}.
\end{multline*}
Since $G_{\infty,\kappa}(t)L^{2}(\mathbb{R}^{d}) \rightarrow D(H_{\infty,\kappa})$, then the Stone theorem (in the kernels sense) provides:
\begin{equation*}
\lim_{\delta t \downarrow 0} \frac{1}{\delta t}\left\{\left(\mathcal{G}_{L,\kappa,\eta}^{(p)}(t+\delta t) \phi\right)(\cdot\,) - \left(\mathcal{G}_{L,\kappa,\eta}^{(p)}(t)\phi\right)(\cdot\,)\right\}
= - \hat{f}_{L,\eta}(\cdot\,) H_{\infty,\kappa} \int_{\mathbb{R}^{d}}  G_{\infty,\kappa}^{(d)}(\cdot\,,\bold{y};t)f_{L,\eta}(\bold{y}) \phi(\bold{y})\,\mathrm{d}\bold{y}.
\end{equation*}
By using similar arguments to treat the contribution coming from $\mathcal{G}_{L,\kappa,\eta}^{(r)}(\cdot\,)$, we therefore obtain:
\begin{equation}
\label{cfdg}
\lim_{\delta t \downarrow 0} \frac{1}{\delta t}\left\{\mathcal{G}_{L,\kappa,\eta}(t+\delta t) \phi - \mathcal{G}_{L,\kappa,\eta}(t)\phi\right\} = - \hat{f}_{L,\eta} H_{\infty,\kappa} G_{\infty,\kappa}(t)f_{L,\eta} \phi - \hat{\hat{f}}_{L,\eta} h_{L,\kappa,\eta} g_{L,\kappa,\eta}(t)f_{L,\eta}^{c}\phi.
\end{equation}
(C)  Let $D_{0}:=\{ \phi \in \mathcal{C}^{1}(\overline{\Lambda_{L}^{d}})\cap \mathcal{C}^{2}(\Lambda_{L}^{d}), \phi\vert_{\partial \Lambda_{L}^{d}}=0, \Delta \phi \in L^{2}(\Lambda_{L}^{d})\}$ be the domain on which $H_{L,\kappa}$ is essentially self-adjoint. In the weak sense for any $\varphi \in D_{0}$, $\psi \in \mathcal{C}_{0}^{\infty}(\Lambda_{L}^{d})$ and $t>0$:
\begin{equation*}
l_{L}(\varphi,\psi) := \left\langle H_{L,\kappa} \varphi, \mathcal{G}_{L,\kappa,\eta}(t) \psi \right\rangle_{L^{2}(\Lambda_{L}^{d})} =- \left\langle \varphi, \frac{\partial\mathcal{G}_{L,\kappa,\eta}}{\partial t}(t) \psi \right\rangle_{L^{2}(\Lambda_{L}^{d})} + \left\langle \varphi, \mathcal{W}_{L,\kappa,\eta}(t) \psi \right\rangle_{L^{2}(\Lambda_{L}^{d})},
\end{equation*}
where $\frac{\partial \mathcal{G}_{L,\kappa,\eta}}{\partial t}(t)$ denotes the operator in the r.h.s. of \eqref{cfdg}. Note that the second equality is obtained by performing some integration by parts, and afterwards by using the following identities:
\begin{equation*}
H_{L,\kappa} \hat{f}_{L,\eta} = H_{\infty,\kappa} \hat{f}_{L,\eta} = \left[H_{\infty,\kappa}, \hat{f}_{L,\eta}\right] + \hat{f}_{L,\eta}H_{\infty,\kappa},
\end{equation*}
as well as (recall that the potential $V_{L}$ in \eqref{ptih} satisfies $V_{L}(\bold{x}) =\vert \bold{x}\vert^{2}$ on $\mathrm{Supp}(\hat{\hat{f}}_{L,\eta})$):
\begin{equation*}
H_{L,\kappa} \hat{\hat{f}}_{L,\eta} = h_{L,\kappa,\eta} \hat{\hat{f}}_{L,\eta} = \left[h_{L,\kappa,\eta}, \hat{\hat{f}}_{L,\eta}\right] + \hat{\hat{f}}_{L,\eta} h_{L,\kappa,\eta}.
\end{equation*}
Since $l_{L}(\varphi,\cdot\,)$ is a bounded linear functional $\forall\varphi \in D_{0}$, then
$\mathcal{C}_{0}^{\infty}(\Lambda_{L}^{d}) \owns \psi \mapsto l_{L}(\varphi,\psi)$ can be extended in a linear and bounded functional on $L^{2}(\Lambda_{L}^{d})$ by the B.L.T. theorem. As well, since $l_{L}(\cdot\,,\psi)$ is a bounded linear functional $\forall \psi \in L^{2}(\Lambda_{L}^{d})$ then $\varphi \mapsto l_{L}(\varphi,\psi)$ can be extended on the self-adjointness domain $D(H_{L,\kappa})$. This means that $\forall t>0$, $\mathrm{Ran}(\mathcal{G}_{L,\kappa,\eta}(t)) \subset D(H_{L,\kappa})$. Hence:
\begin{equation*}
\left\langle \varphi, H_{L,\kappa} \mathcal{G}_{L,\kappa,\eta}(t) \psi \right\rangle_{L^{2}(\Lambda_{L}^{d})} = - \left\langle \varphi, \frac{\partial\mathcal{G}_{L,\kappa,\eta}}{\partial t}(t) \psi \right\rangle_{L^{2}(\Lambda_{L}^{d})} + \left\langle \varphi, \mathcal{W}_{L,\kappa,\eta}(t) \psi \right\rangle_{L^{2}(\Lambda_{L}^{d})}.
\end{equation*}
Finally, from \eqref{normcalR} $\Vert \mathcal{W}_{L,\kappa,\eta}(t)\Vert \leq C t^{-\frac{1}{2}}$ $\forall 0<t\leq 1$. Hence $\Vert\mathcal{W}_{L,\kappa,\eta}(t)\Vert$ is integrable in $t \sim 0$.\qed

\subsubsection{Proof of intermediary results.}
\label{intermB}

\noindent \textit{Proof of Lemma \ref{lem7}.} \eqref{kerhL} follows from the Feynman-Kac formula in
\cite[Thm. X.68]{RS2} together with \eqref{fondineq} and the definition of the $L_{0}$ in \eqref{defR0} leading to $(L-L^{\eta})^{2} \geq L^{2}/2$ $\forall L \in [L_{0}(\eta),\infty)$. Next, let us turn to the proof of \eqref{majgradh}-\eqref{majlaph0}. To do that, let us introduce an operator of reference. $\forall d \in \{1,2,3\}$, $\forall 0<\eta<1$, $\forall L \in (0,\infty)$ and $\forall \kappa>0$, define on $\mathcal{C}_{0}^{\infty}(\Lambda_{L}^{d})$:
\begin{equation}
\label{H2}
\tilde{h}_{L,\kappa,\eta} := \frac{1}{2}\left(-i\nabla_{\bold{x}}\right)^{2} + \frac{\kappa^{2}}{2} \tilde{V}_{L,\eta}(\bold{x}),\quad \tilde{V}_{L,\eta}(\bold{x}):= \frac{1}{4} \left(L - L^{\eta}\right)^{2}.
\end{equation}
By standard arguments, \eqref{H2} extends to a family of self-adjoint and semi-bounded operators $\forall L \in (0,\infty)$, denoted again by
$\tilde{h}_{L,\kappa,\eta}$. For any $0<\eta<1$, $L \in (0,\infty)$ and $\kappa>0$, let $\{\tilde{g}_{L,\kappa,\eta}(t):= \mathrm{e}^{- t \tilde{h}_{L,\kappa,\eta}}:L^{2}(\Lambda_{L}^{d}) \rightarrow L^{2}(\Lambda_{L}^{d})\}_{t \geq 0}$ be the strongly-continuous one-parameter semigroup generated by $\tilde{h}_{L,\kappa,\eta}$. Its integral kernel denoted by $\tilde{g}_{L,\kappa,\eta}^{(d)}$ is explicitly known and reads as:
\begin{equation}
\label{exprtildg}
\forall (\bold{x},\bold{y}) \in \Lambda_{L}^{2d},\,\forall t>0,\quad \tilde{g}_{L,\kappa,\eta}^{(d)}(\bold{x},\bold{y};t) = \mathrm{e}^{-\frac{\kappa^{2}}{8} \left(L - L^{\eta}\right)^{2}t} G_{L,0}^{(d)}(\bold{x},\bold{y};t),
\end{equation}
where $G_{L,0}^{(d)}$ is the kernel of the semigroup generated by the Dirichlet Laplacian in $L^{2}(\Lambda_{L}^{d})$, see \eqref{kerdirid}. Note that \eqref{exprtildg} directly  follows from the Feynman-Kac formula. The starting-point of the proof of \eqref{majgradh}-\eqref{majlaph0} is a Duhamel-like formula to express the semigroup $\{g_{L,\kappa,\eta}(t)\}_{t>0}$ in terms of $\{\tilde{g}_{L,\kappa,\eta}(t)\}_{t>0}$ whose integral kernel is given in \eqref{exprtildg}. Let $0<\eta<1$, $L \in [L_{0}(\eta),\infty)$ (see \eqref{defR0}) and $\kappa>0$ be fixed. In the bounded operators sense on $L^{2}(\Lambda_{L}^{d})$, it takes place:
\begin{equation}
\label{du}
\forall t>0,\quad g_{L,\kappa,\eta}(t) = \tilde{g}_{L,\kappa,\eta}(t) - \int_{0}^{t}  \tilde{g}_{L,\kappa,\eta}(s)\left\{h_{L,\kappa,\eta} - \tilde{h}_{L,\kappa,\eta}\right\}  g_{L,\kappa,\eta}(t-s)\,\mathrm{d}s,
\end{equation}
where we used the self-adjointness of the semigroups $\{g_{L,\kappa,\eta}(t)\}_{t\geq 0}$, $\{\tilde{g}_{L,\kappa,\eta}(t)\}_{t\geq 0}$.\\
\textit{Proof of \eqref{majgradh}}. From \eqref{du}, it follows in the kernels sense:
\begin{gather}
\label{kerdum}
\forall (\bold{x},\bold{y}) \in \Lambda_{L}^{2d},\,\forall t>0,\quad \nabla_{\bold{x}} g_{L,\kappa,\eta}^{(d)}(\bold{x},\bold{y};t) = \nabla_{\bold{x}} \tilde{g}_{L,\kappa,\eta}^{(d)}(\bold{x},\bold{y};t) - \frac{1}{2} \mathfrak{q}_{L,\kappa,\eta}^{(d)}(\bold{x},\bold{y};t),\\
\mathfrak{q}_{L,\kappa,\eta}^{(d)}(\bold{x},\bold{y};t) := \kappa^{2} \int_{0}^{t} \int_{\Lambda_{L}^{d}}  \nabla_{\bold{x}} \tilde{g}_{L,\kappa,\eta}^{(d)}(\bold{x},\bold{z};s)\left\{V_{L,\eta}(\bold{z}) - \tilde{V}_{L,\eta}(\bold{z})\right\} g_{L,\kappa,\eta}^{(d)}(\bold{z},\bold{y};t-s)\, \mathrm{d}\bold{z}\,\mathrm{d}s.\nonumber
\end{gather}
Recall that $V_{L,\eta}(\bold{z}) - \tilde{V}_{L,\eta}(\bold{z}) = \vert \bold{z}\vert^{2}-\frac{1}{4}(L-L^{\eta})^{2}$ on $\mathrm{Supp}(\hat{\hat{f}}_{L,\eta})$, $0$ otherwise. Let us estimate the first kernel in the r.h.s. of \eqref{kerdum}. From \eqref{derkL} and \eqref{exprtildg}, there exists a constant $C_{d}>0$ s.t.
\begin{equation}
\label{dertildd}
\forall(\bold{x},\bold{y}) \in \Lambda_{L}^{2d},\,\forall t>0,\quad \left\vert \nabla_{\bold{x}} \tilde{g}_{L,\kappa,\eta}^{(d)}(\bold{x},\bold{y};t) \right\vert \leq C_{d} \frac{(1+t)^{d}}{\sqrt{t}} \mathrm{e}^{-  \frac{\kappa^{2}}{4} \frac{L^{2}}{4} t} G_{\infty,0}^{(d)}(\bold{x},\bold{y};2t).
\end{equation}
Subsequently, from \eqref{dertildd} along with \eqref{kerhL}, there exists another constant $C_{d}>0$ s.t.
\begin{equation*}
\forall(\bold{x},\bold{y}) \in \Lambda_{L}^{2d},\,\forall t>0,\quad \left\vert \mathfrak{q}_{L,\kappa,\eta}^{(d)}(\bold{x},\bold{y};t)\right\vert \leq C_{d} \kappa^{2} L^{2} (1+t)^{d} \mathrm{e}^{-\frac{\kappa^{2}}{4} \frac{L^{2}}{4} t} G_{\infty,0}^{(d)}(\bold{x},\bold{y};2t) \int_{0}^{t}\frac{\mathrm{d}s}{\sqrt{s}},
\end{equation*}
where we used in the last inequality \eqref{propsemi}. Finally use \eqref{redexp} to get rid of the $L^{2}$ what leads to:
\begin{equation}
\label{Tcalfi}
\left\vert \mathfrak{q}_{L,\kappa,\eta}^{(d)}(\bold{x},\bold{y};t)\right\vert \leq C_{d} \frac{(1+t)^{d}}{\sqrt{t}} \mathrm{e}^{-\frac{\kappa^{2}}{8} \frac{L^{2}}{4} t} G_{\infty,0}^{(d)}(\bold{x},\bold{y};2t),
\end{equation}
for another constant $C_{d}>0$. It remains to gather \eqref{dertildd} and \eqref{Tcalfi} together.\\
\textit{Proof of \eqref{majlaph0}}. Starting from the below identity which holds in the bounded operators sense:
\begin{equation*}
\forall t>0,\quad \left[\left(-i\nabla\right),\tilde{g}_{L,\kappa,\eta}(t)\right] = - \int_{0}^{t}  \tilde{g}_{L,\kappa,\eta}(t-s)\left[\left(-i\nabla\right),\tilde{h}_{L,\kappa,\eta}\right] \tilde{g}_{L,\kappa,\eta}(s)\,\mathrm{d}s,
\end{equation*}
then by using that $[(-i\nabla),\tilde{h}_{L,\kappa,\eta}] =0$, one gets from \eqref{du} on $L^{2}(\Lambda_{L}^{d})$:
\begin{equation*}
\forall t>0,\quad \left(-i\nabla\right) g_{L,\kappa,\eta}(t) = \left(-i\nabla\right) \tilde{g}_{L,\kappa,\eta}(t) - \int_{0}^{t}  \tilde{g}_{L,\kappa,\eta}(s)\left(-i \nabla\right)\left\{h_{L,\kappa,\eta} - \tilde{h}_{L,\kappa,\eta}\right\} g_{L,\kappa,\eta}(t-s)\,\mathrm{d}s.
\end{equation*}
It follows in the kernels sense:
\begin{gather}
\forall(\bold{x},\bold{y}) \in \Lambda_{L}^{2d},\forall t>0,\quad \Delta_{\bold{x}} g_{L,\kappa,\eta}^{(d)}(\bold{x},\bold{y};t) = \Delta_{\bold{x}}\tilde{g}_{L,\kappa,\eta}^{(d)}(\bold{x},\bold{y};t) - \frac{1}{2} \sum_{l=1}^{2} \mathfrak{u}_{L,\kappa,\eta}^{(d), l}(\bold{x},\bold{y};t),\nonumber\\
\mathfrak{u}_{L,\kappa,\eta}^{(d),1}(\bold{x},\bold{y};t) :=  \kappa^{2} \int_{0}^{t}  \int_{\Lambda_{L}^{d}}  \nabla_{\bold{x}} \tilde{g}_{L,\kappa,\eta}^{(d)}(\bold{x},\bold{z};s) \left(\nabla_{\bold{z}} V_{L,\eta}\right)(\bold{z}) g_{L,\kappa,\eta}^{(d)}(\bold{z},\bold{y};t-s)\,\mathrm{d}\bold{z}\,\mathrm{d}s,\nonumber\\
\mathfrak{u}_{L,\kappa,\eta}^{(d),2}(\bold{x},\bold{y};t) := \kappa^{2} \int_{0}^{t}  \int_{\Lambda_{L}^{d}} \nabla_{\bold{x}} \tilde{g}_{L,\kappa,\eta}^{(d)}(\bold{x},\bold{z};s)\left\{\tilde{V}_{L,\eta}(\bold{z}) - V_{L,\eta}(\bold{z})\right\} \nabla_{\bold{z}} g_{L,\kappa,\eta}^{(d)}(\bold{z},\bold{y};t-s)\, \mathrm{d}\bold{z}\,\mathrm{d}s.\nonumber
\end{gather}
From \eqref{lapkL} and \eqref{exprtildg}, there exists a constant $C_{d}>0$ s.t.
\begin{equation*}
\label{dertildd2}
\forall(\bold{x},\bold{y}) \in \Lambda_{L}^{2d},\forall t>0,\quad \left\vert \Delta_{\bold{x}} \tilde{g}_{L,\kappa,\eta}^{(d)}(\bold{x},\bold{y};t)\right\vert \leq C_{d} \frac{(1+t)^{d}}{t} \mathrm{e}^{-  \frac{\kappa^{2}}{4} \frac{L^{2}}{4} t} G_{\infty,0}^{(d)}(\bold{x},\bold{y};2t).
\end{equation*}
Subsequently, by mimicking the method leading to \eqref{Tcalfi}, there exists another $C_{d}>0$ s.t.
\begin{equation*}
\forall(\bold{x},\bold{y}) \in \Lambda_{L}^{2d},\,\forall t>0,\quad \left\vert \mathfrak{u}_{L,\kappa,\eta}^{(d),1}(\bold{x},\bold{y};t) \right\vert \leq C_{d} \frac{(1+t)^{d}}{\sqrt{t}} \mathrm{e}^{-\frac{\kappa^{2}}{8} \frac{L^{2}}{4} t} G_{\infty,0}^{(d)}(\bold{x},\bold{y};2t).
\end{equation*}
By the same method again but replacing the estimate \eqref{kerhL} with \eqref{majgradh}, we have:
\begin{equation*}
\left\vert \mathfrak{u}_{L,\kappa,\eta}^{(d),2}(\bold{x},\bold{y};t) \right\vert \leq C_{d} \frac{(1+t)^{2d}}{t} \mathrm{e}^{-\frac{\kappa^{2}}{16} \frac{L^{2}}{4} t} G_{\infty,0}^{(d)}(\bold{x},\bold{y};2t) \int_{0}^{t}\frac{\mathrm{d}s}{\sqrt{s}\sqrt{t-s}}.
\end{equation*}
Gathering the three above estimates together, then the proof of \eqref{majlaph0} is over.\qed\\

\noindent \textit{Proof of Lemma \ref{toues}.} Let $d \in \{1,2,3\}$, $0<\eta<1$, $L \in [L_{0}(\eta),\infty)$ and $\kappa>0$ kept fixed.\\
$\mathrm{(i)}$. From \eqref{calGad} written in the kernels sense, then $\forall(\bold{x},\bold{y}) \in \Lambda_{L}^{2d}$ and $\forall t>0$:
\begin{equation*}
\begin{split}
\nabla_{\bold{x}}\left(\mathcal{G}_{L,\kappa,\eta}^{*}\right)^{(d)}(\bold{x},\bold{y};t) &= \left(\nabla f_{L,\eta}\right)(\bold{x}) G_{\infty,\kappa}^{(d)}(\bold{x},\bold{y};t) \hat{f}_{L,\eta}(\bold{y}) + f_{L,\eta}(\bold{x}) \nabla_{\bold{x}}G_{\infty,\kappa}^{(d)}(\bold{x},\bold{y};t) \hat{f}_{L,\eta}(\bold{y}) + \\
&+ \left(\nabla f_{L,\eta}^{c}\right)(\bold{x}) g_{L,\kappa,\eta}^{(d)}(\bold{x},\bold{y};t) \hat{\hat{f}}_{L,\eta}(\bold{y}) + f_{L,\eta}^{c}(\bold{x}) \nabla_{\bold{x}} g_{L,\kappa,\eta}^{(d)}(\bold{x},\bold{y};t) \hat{\hat{f}}_{L,\eta}(\bold{y}).
\end{split}
\end{equation*}
\eqref{kerud} is an upper bound for the two first kernels in the above r.h.s. obtained from \eqref{Mehler}-\eqref{multd} and \eqref{dermelh}. \eqref{kervd} is an upper bound for the two last kernels obtained from \eqref{kerhL} and \eqref{majgradh}.\\
$\mathrm{(ii)}$. From \eqref{calRad} written in the kernels sense, then $\forall(\bold{x},\bold{y}) \in \Lambda_{L}^{2d}$ and $\forall t>0$:
\begin{gather}
\nabla_{\bold{x}}\left(\mathcal{W}_{L,\kappa,\eta}^{*}\right)^{(d)}(\bold{x},\bold{y};t) = \sum_{m=1}^{4} Q_{L,\kappa,\eta}^{(d),m}(\bold{x},\bold{y};t), \quad \textrm{with:} \nonumber\\
\label{Q1}
\begin{split}
&Q_{L,\kappa,\eta}^{(d),1}(\bold{x},\bold{y};t):=  - i\left(\nabla f_{L,\eta}\right)(\bold{x}) \left[\left(-i \nabla\right),G_{\infty,\kappa}(t)\right](\bold{x},\bold{y}) \left(\nabla \hat{f}_{L,\eta}\right)(\bold{y}) + \\
&+ \left(\nabla f_{L,\eta}\right)(\bold{x}) \nabla_{\bold{x}} G_{\infty,\kappa}^{(d)}(\bold{x},\bold{y};t)\left(\nabla \hat{f}_{L,\eta}\right)(\bold{y}) - \left(\nabla f_{L,\eta}\right)(\bold{x}) G_{\infty,\kappa}^{(d)}(\bold{x},\bold{y};t) \frac{1}{2} \left(\Delta \hat{f}_{L,\eta}\right)(\bold{y}),
\end{split}\\
\label{Q2}
\begin{split}
&Q_{L,\kappa,\eta}^{(d),2}(\bold{x},\bold{y};t):=  - i f_{L,\eta}(\bold{x}) \nabla_{\bold{x}} \left[\left(-i \nabla\right),G_{\infty,\kappa}(t)\right](\bold{x},\bold{y}) \left(\nabla \hat{f}_{L,\eta}\right)(\bold{y}) + \\
&+ f_{L,\eta}(\bold{x}) \Delta_{\bold{x}} G_{\infty,\kappa}^{(d)}(\bold{x},\bold{y};t) \left(\nabla \hat{f}_{L,\eta}\right)(\bold{y}) - f_{L,\eta}(\bold{x}) \nabla_{\bold{x}} G_{\infty,\kappa}^{(d)}(\bold{x},\bold{y};t) \frac{1}{2} \left(\Delta \hat{f}_{L,\eta}\right)(\bold{y}),
\end{split}\\
\label{Q3}
\begin{split}
&Q_{L,\kappa,\eta}^{(d),3}(\bold{x},\bold{y};t):= - i\left(\nabla f_{L,\eta}^{c}\right)(\bold{x}) \left[\left(-i \nabla\right),g_{L,\kappa,\eta}(t)\right](\bold{x},\bold{y}) \left(\nabla \hat{\hat{f}}_{L,\eta}\right)(\bold{y})+ \\
&+ \left(\nabla f_{L,\eta}^{c}\right)(\bold{x}) \nabla_{\bold{x}} g_{L,\kappa,\eta}^{(d)}(\bold{x},\bold{y};t)\left(\nabla \hat{\hat{f}}_{L,\eta}\right)(\bold{y}) - \left(\nabla f_{L,\eta}^{c}\right)(\bold{x}) g_{L,\kappa,\eta}^{(d)}(\bold{x},\bold{y};t) \frac{1}{2} \left(\Delta \hat{\hat{f}}_{L,\eta}\right)(\bold{y}),
\end{split}\\
\label{Q4}
\begin{split}
&Q_{L,\kappa,\eta}^{(d),4}(\bold{x},\bold{y};t):=  - i f_{L,\eta}^{c}(\bold{x}) \nabla_{\bold{x}} \left[\left(-i \nabla\right),g_{L,\kappa,\eta}(t)\right](\bold{x},\bold{y}) \left(\nabla \hat{\hat{f}}_{L,\eta}\right)(\bold{y}) + \\
&+ f_{L,\eta}^{c}(\bold{x}) \Delta_{\bold{x}} g_{L,\kappa,\eta}^{(d)}(\bold{x},\bold{y};t) \left(\nabla \hat{\hat{f}}_{L,\eta}\right)(\bold{y}) - f_{L,\eta}^{c}(\bold{x}) \nabla_{\bold{x}} g_{L,\kappa,\eta}^{(d)}(\bold{x},\bold{y};t) \frac{1}{2} \left(\Delta \hat{\hat{f}}_{L,\eta}\right)(\bold{y}).
\end{split}
\end{gather}
Let us first estimate \eqref{Q1}. In view of \eqref{dermelh}, \eqref{kerud} is clearly an upper bound for the two last terms in the r.h.s. of \eqref{Q1}. For the first term in \eqref{Q1}, we use \eqref{commut1} in the kernels sense. Then, there exists a constant $C_{d}>0$ s.t. $\forall(\bold{x},\bold{y}) \in \Lambda_{L}^{2d}$ and $\forall t>0$:
\begin{equation}
\label{supimp1}
\frac{\kappa^{2}}{2} \int_{0}^{t}  \int_{\mathbb{R}^{d}}  G_{\infty,\kappa}^{(d)}(\bold{x},\bold{z};t-s) \left(\nabla_{\bold{z}} \vert \bold{z}\vert^{2}\right) G_{\infty,\kappa}^{(d)}(\bold{z},\bold{y};s) \,\mathrm{d}\bold{z}\,\mathrm{d}s \leq C_{d} (\kappa^{2} L + \kappa^{\frac{3}{2}}) t G_{\infty,\kappa}^{(d)}(\bold{x},\bold{y};t,2).
\end{equation}
Here, we used that $\vert \bold{z}\vert \leq \vert \bold{x}-\bold{z}\vert + \vert \bold{x}\vert$, then \eqref{redexp} to get rid of the factor $\vert \bold{x}-\bold{z}\vert$ and \eqref{propsemi2}, and finally the lower bound $\coth(\alpha) \geq 1$ $\forall \alpha \geq 0$. Next, we use the property \eqref{disjsup2} to get rid of the powers of $\kappa$ in \eqref{supimp1} via \eqref{redexp}. Hence, there exist two other constants $c, C_{d}>0$ s.t. on $\Lambda_{L}^{2d}$:
\begin{multline*}
\left\vert i\left(\nabla f_{L,\eta}\right)(\bold{x}) \left[\left(-i \nabla\right),G_{\infty,\kappa}(t)\right](\bold{x},\bold{y}) \left(\nabla \hat{f}_{L,\eta}\right)(\bold{y}) \right\vert \\
\leq C_{d}L^{-3\eta} (1+ L^{1-\eta}) t \mathrm{e}^{-c \kappa L^{2\eta} \coth\left(\frac{\kappa}{2}t\right)} G_{\infty,\kappa}^{(d)}(\bold{x},\bold{y};t,4).
\end{multline*}
Restricting to $1>\eta>\frac{1}{4}$, and gathering the above estimate with \eqref{kerud} together, then there exist two other constants $c,C_{d}>0$ s.t. $\forall L \in [L_{0}(\eta),\infty)$, $\forall(\bold{x},\bold{y}) \in \Lambda_{L}^{2d}$ and $\forall t>0$:
\begin{equation}
\label{majQ1}
\left\vert Q_{L,\kappa,\eta}^{(d),1}(\bold{x},\bold{y};t)\right\vert \leq C_{d}\left(1+\sqrt{\kappa}\right) \sqrt{\coth\left(\frac{\kappa}{2} t\right)} (1+t) \mathrm{e}^{-c \kappa L^{2\eta} \coth\left(\frac{\kappa}{2}t\right)} G_{\infty,\kappa}^{(d)}(\bold{x},\bold{y};t,4).
\end{equation}
Subsequently, let us turn to \eqref{Q2}. From \eqref{dermelh} and \eqref{lapmelh} together with the property \eqref{disjsup1}, then there exist two other constants $c , C_{d}>0$ s.t. $\forall (\bold{x},\bold{y}) \in \Lambda_{L}^{2d}$ and $\forall t>0$:
\begin{multline*}
\left\vert f_{L,\eta}(\bold{x}) \nabla_{\bold{x}} G_{\infty,\kappa}^{(d)}(\bold{x},\bold{y};t) \frac{1}{2} \left(\Delta \hat{f}_{L,\eta}\right)(\bold{y})  + f_{L,\eta}(\bold{x}) \Delta_{\bold{x}} G_{\infty,\kappa}^{(d)}(\bold{x},\bold{y};t) \left(\nabla \hat{f}_{L,\eta}\right)(\bold{y})\right\vert \\
\leq C_{d} \sqrt{\kappa} \sqrt{\coth\left(\frac{\kappa}{2} t\right)} \mathrm{e}^{- c \kappa L^{2\eta} \coth\left(\frac{\kappa}{2} t\right)} G_{\infty,\kappa}^{(d)}(\bold{x},\bold{y};t,4).
\end{multline*}
Here, the property \eqref{disjsup1} is essential to remove a $\sqrt{\coth(\kappa t)}$ in the numerator of \eqref{lapmelh}. For the first term of \eqref{Q2}, we use the same reasoning than the one leading to \eqref{supimp1} combined with the property \eqref{disjsup1}. Thus, there exist two other constants $c,C_{d} >0$ s.t. $\forall (\bold{x},\bold{y}) \in \Lambda_{L}^{2d}$ and $\forall t>0$:
\begin{multline*}
\frac{\kappa^{2}}{2} \left\vert f_{L,\eta}(\bold{x}) \int_{0}^{t}  \int_{\mathbb{R}^{d}}  \nabla_{\bold{x}} G_{\infty,\kappa}^{(d)}(\bold{x},\bold{z};t-s) \left(\nabla_{\bold{z}} \vert \bold{z}\vert^{2}\right) G_{\infty,\kappa}^{(d)}(\bold{z},\bold{y};s) \left(\nabla \hat{f}_{L,\eta}\right)(\bold{y})\,\mathrm{d}\bold{z}\,\mathrm{d}s \right\vert \\
\leq C_{d} \left(1+\sqrt{\kappa}\right) L^{-4\eta}(1+L) (1+t) \mathrm{e}^{- c \kappa L^{2\eta} \coth\left(\frac{\kappa}{2} t\right)} G_{\infty,\kappa}^{(d)}(\bold{x},\bold{y};t,8).
\end{multline*}
Restricting to $1>\eta>\frac{1}{4}$, and gathering the above estimates together, then there exist two other constants $c,C_{d}>0$ s.t. $\forall L \in [L_{0}(\eta),\infty)$, $\forall (\bold{x},\bold{y}) \in \Lambda_{L}^{2d}$ and $\forall t>0$:
\begin{equation}
\label{majQ2}
\left\vert Q_{L,\kappa,\eta}^{(d),2}(\bold{x},\bold{y};t)\right\vert \leq C_{d} \left(1+\sqrt{\kappa}\right) \sqrt{\coth\left(\frac{\kappa}{2} t\right)} (1+t) \mathrm{e}^{-c \kappa L^{2\eta} \coth\left(\frac{\kappa}{2}t\right)} G_{\infty,\kappa}^{(d)}(\bold{x},\bold{y};t,8).
\end{equation}
The estimate in \eqref{kerwd1} follows by adding \eqref{majQ1} and \eqref{majQ2} together.\\
We continue with \eqref{Q3}. \eqref{kervd} is an upper bound for the last two terms in the r.h.s. of \eqref{Q3}. From \eqref{commut2} in the kernels sense, then by \eqref{redexp} there exist two other constants $c,C_{d}>0$ s.t.
\begin{equation*}
\frac{\kappa^{2}}{2} \int_{0}^{t}  \int_{\Lambda_{L}^{d}}  g_{L,\kappa,\eta}^{(d)}(\bold{x},\bold{z};t-s) \left(\nabla_{\bold{z}} V_{L,\eta}\right)(\bold{z}) g_{L,\kappa,\eta}^{(d)}(\bold{z},\bold{y};s)\,\mathrm{d}\bold{z}\,\mathrm{d}s
\leq C_{d} \mathrm{e}^{-\frac{\kappa^{2}}{8}\frac{L^{2}}{4} t} \mathrm{e}^{-c \frac{L^{2\eta}}{t}} G_{\infty,0}^{(d)}(\bold{x},\bold{y};2t).
\end{equation*}
We conclude that there exist two other constants $c,C_{d}>0$ s.t. $\forall(\bold{x},\bold{y}) \in \Lambda_{L}^{2d}$:
\begin{equation}
\label{majQ3}
\forall t>0,\quad \left\vert Q_{L,\kappa,\eta}^{(d),3}(\bold{x},\bold{y};t)\right\vert \leq C_{d} \frac{(1+t)^{d}}{\sqrt{t}} \mathrm{e}^{-\frac{\kappa^{2}}{8} \frac{L^{2}}{4} t} \mathrm{e}^{-c \frac{L^{2\eta}}{t}} G_{\infty,0}^{(d)}(\bold{x},\bold{y};4t).
\end{equation}
Turning to \eqref{Q4}, one can prove that there exist two other constants $c,C_{d} >0$ s.t. on $\Lambda_{L}^{2d}$:
\begin{equation}
\label{majQ4}
\forall t>0,\quad \left\vert Q_{L,\kappa,\eta}^{(d),4}(\bold{x},\bold{y};t)\right\vert \leq C_{d} \frac{(1+t)^{2d}}{\sqrt{t}} \mathrm{e}^{-\frac{\kappa^{2}}{16} \frac{L^{2}}{4} t} \mathrm{e}^{-c \frac{L^{2\eta}}{t}} G_{\infty,0}^{(d)}(\bold{x},\bold{y};4t).
\end{equation}
Here, we used \eqref{disjsup2} combined with \eqref{redexp} to get rid of a $\sqrt{t}$ in the denominator of \eqref{majlaph0}. The estimate in \eqref{kerwd2} follows by adding \eqref{majQ3} and \eqref{majQ4} together, then by taking into account the support of the cutoff functions introduced in Sec. \ref{GPTo}. \qed

\section{Acknowledgments.}

B.S. was partially supported by the Lundbeck Foundation, and the European Research Council under the European Community's Seventh Framework Program (FP7/2007--2013)/ERC grant agreement 202859. A part of this work was done while the second author was visiting DIAS-STP (Dublin), B.S. is grateful for invitation and financial support. Both authors warmly thank Horia Cornean, Tony Dorlas and Cyril Levy for helpful and stimulating discussions.

\section{Appendix--The semigroup: A review of some properties.}
\label{Appen1}

Here, we collect the technical results we use throughout the paper involving the semigroup generated by $H_{L,\kappa}$, see Sec. \ref{defop}.\\
\indent For simplicity's sake, we hereafter use the notation $\Lambda_{\infty}:= \mathbb{R}$. From \eqref{HL}-\eqref{Hinfini}, recall that:
\begin{equation}
\label{HLinf}
\forall L \in (0,\infty],\quad H_{L,\kappa} = \frac{1}{2}\left(-i \nabla_{\bold{x}}\right)^{2} + \frac{1}{2} \kappa^{2} \vert \bold{x}\vert^{2}\quad \textrm{in $L^{2}(\Lambda_{L}^{d})$,\, $d\in\{1,2,3\}$.}
\end{equation}
Below, we allow the value $\kappa=0$; in that case, $H_{L,0}$ with $L<\infty$ is nothing but the Dirichlet Laplacian and $H_{\infty,0}$ the free Laplacian on the whole space whose self-adjointness domain is $W^{2,2}(\Lambda_{\infty}^{d})$.\\
\indent Recall some properties on the strongly continuous one-parameter semigroup $\{G_{L,\kappa}(t) := \mathrm{e}^{- t H_{L,\kappa}} : L^{2}(\Lambda_{L}^{d}) \rightarrow L^{2}(\Lambda_{L}^{d})\}_{t \geq 0}$ generated by $H_{L,\kappa}$ in \eqref{HLinf}. We refer to \cite[Sec. X.8]{RS2} and \cite[Sec. B]{Si}. As already mentioned, $\forall \kappa\geq0$ and $\forall L \in (0,\infty]$ it is a self-adjoint and positive operator on $L^{2}(\Lambda_{L}^{d})$ by the spectral theorem and the functional
calculus. Moreover, since $\{G_{L,\kappa}(t)\}_{t>0}$ is bounded from $L^{2}(\Lambda_{L}^{d})$ to $L^{\infty}(\Lambda_{L}^{d})$, then it is an integral operator by the Dunford-Gelfand-Pettis theorem.\\
\indent Let us turn to the integral kernel of $\{G_{L,\kappa}(t)\}_{t>0}$ we denote by $G_{L,\kappa}^{(d)}$. $\forall \kappa\geq0$ and $\forall L \in (0,\infty]$, $G_{L,\kappa}^{(d)}$ is jointly continuous in $(\bold{x},\bold{y},t) \in \overline{\Lambda_{L}^{d}}\times \overline{\Lambda_{L}^{d}} \times (0,\infty)$ and vanishes if $\bold{x} \in \partial \Lambda_{L}^{d}$ or $\bold{y} \in \partial \Lambda_{L}^{d}$. When $L=\infty$, it is explicitly known. If $\kappa=0$, it is the so-called heat kernel reading for $d=1$ as:
\begin{equation}
\label{heatk}
\forall(x,y)\in \Lambda_{\infty}^{2},\,\forall t>0,\quad G_{\infty,0}^{(d=1)}(x,y;t) := \frac{1}{\sqrt{2\pi}} \frac{\mathrm{e}^{-\frac{(x-y)^{2}}{2t}}}{\sqrt{t}}.
\end{equation}
If $\kappa>0$, the one-dimensional kernel is given by the so-called Mehler formula, see
\cite[pp. 176]{Ku}:
\begin{equation}
\label{Mehler}
\forall(x,y)\in \Lambda_{\infty}^{2},\,\forall t>0,\quad G_{\infty,\kappa}^{(d=1)}(x,y;t)  = \sqrt{\frac{\kappa}{2 \pi \sinh(\kappa t)}} \mathrm{e}^{-\frac{\kappa}{4} \left[(x+y)^{2} \tanh\left(\frac{\kappa}{2} t\right) + (x-y)^{2} \coth\left(\frac{\kappa}{2} t\right)\right]}.
\end{equation}
Note that the multidimensional kernel (i.e. $d=2,3$) is directly obtained from \eqref{heatk} or \eqref{Mehler} by:
\begin{equation}
\label{multd}
\forall \kappa\geq 0,\quad G_{\infty,\kappa}^{(d)}(\bold{x},\bold{y};t) := \prod_{j=1}^{d} G_{\infty,\kappa}^{(d=1)}(x_{j},y_{j};t), \quad \bold{x}:=\{x_{j}\}_{j=1}^{d},\, \bold{y}:=\{y_{j}\}_{j=1}^{d}.
\end{equation}
When restricting to $L \in (0,\infty)$, the mapping $L \mapsto G_{L,\kappa}^{(d)}(\bold{x},\bold{y};t)$ is positive and monotone increasing. This leads to the following pointwise inequality which holds $\forall\kappa\geq 0$ and $\forall L\in (0,\infty)$:
\begin{equation}
\label{fondineq}
\forall (\bold{x},\bold{y},t) \in \overline{\Lambda_{L}^{d}}\times \overline{\Lambda_{L}^{d}}\times(0,\infty),\quad G_{L,\kappa}^{(d)}(\bold{x},\bold{y};t) \leq \sup_{L>0} G_{L,\kappa}^{(d)}(\bold{x},\bold{y};t) = G_{\infty,\kappa}^{(d)}(\bold{x},\bold{y};t).
\end{equation}
We mention that, if $\kappa=0$, the kernel $G_{L,0}^{(d)}$ is explicitly known and reads as, see
\cite[Eq. (4.13)]{C1}:
\begin{gather}
\label{kerdirid}
\forall(\bold{x},\bold{y}) \in \Lambda_{L}^{2d},\,\forall t>0,\quad G_{L,0}^{(d)}(\bold{x},\bold{y};t) = \prod_{j=1}^{d} G_{L,0}^{(d=1)}(x_{j},y_{j};t),\\
G_{L,0}^{(d=1)}(x,y;t) := \frac{1}{\sqrt{2t}} \sum_{m \in \mathbb{Z}}\left\{\exp\left(-\frac{(x-y+2mL)^2}{2t}\right) - \exp\left(-\frac{(x+y-2mL-L)^2}{2t}\right)\right\}.\nonumber
\end{gather}

In view of \eqref{Mehler}-\eqref{multd}, let us introduce $\forall \kappa>0$ the new notation:
\begin{equation}
\label{ginfty}
\forall \gamma >0,\quad G_{\infty,\kappa}^{(d)}(\bold{x},\bold{y};t,\gamma) := \left(\frac{\kappa}{2\pi \sinh(\kappa t)}\right)^{\frac{d}{2}} \prod_{j=1}^{d} \mathrm{e}^{-\frac{\kappa}{4 \gamma}\left[(x_{j}+y_{j})^{2} \tanh\left(\frac{\kappa}{2}t\right) + (x_{j}-y_{j})^{2} \coth\left(\frac{\kappa}{2}t\right)\right]},
\end{equation}
with the convention: $G_{\infty,\kappa}^{(d)}(\cdot\,,\cdot\,;t)= G_{\infty,\kappa}^{(d)}(\cdot\,,\cdot\,;t,1)$. Here are collected all the needed estimates:

\begin{lema}
\label{plentes}
$\forall d \in \{1,2,3\}$, there exists a constant $C_{d}>0$ s.t. \\
$\mathrm{(i)}$. $\forall \kappa>0$, $\forall \gamma>0$, $\forall(\bold{x},\bold{y}) \in \Lambda_{\infty}^{2d}$ and $\forall t>0$:
\begin{gather}
\label{roughes}
G_{\infty,\kappa}^{(d)}(\bold{x},\bold{y};t,\gamma) \leq \left(\frac{\kappa}{\sinh(\kappa t)}\right)^{\frac{d}{2}} t^{\frac{d}{2}} \gamma^{\frac{d}{2}} G_{\infty,0}^{(d)}(\bold{x},\bold{y};\gamma t) \leq
\gamma^{\frac{d}{2}} G_{\infty,0}^{(d)}(\bold{x},\bold{y};\gamma t) \leq (2 \pi t)^{-\frac{d}{2}},\\
\label{dermelh}
\left\vert \nabla_{\bold{x}} G_{\infty,\kappa}^{(d)}(\bold{x},\bold{y};t)\right\vert \leq
C_{d} \sqrt{\kappa} \sqrt{\coth\left(\frac{\kappa}{2} t\right)}G_{\infty,\kappa}^{(d)}(\bold{x},\bold{y};t,2),\\
\label{lapmelh}
\left\vert \Delta_{\bold{x}} G_{\infty,\kappa}^{(d)}(\bold{x},\bold{y};t)\right\vert \leq
C_{d} \kappa \coth(\kappa t) G_{\infty,\kappa}^{(d)}(\bold{x},\bold{y};t,2).
\end{gather}
$\mathrm{(ii)}$. $\forall L \in (0,\infty)$, $\forall(\bold{x},\bold{y}) \in \Lambda_{L}^{2d}$ and $\forall t>0$:
\begin{gather}
\label{derkL}
\left\vert \nabla_{\bold{x}} G_{L,0}^{(d)}(\bold{x},\bold{y};t)\right\vert \leq
C_{d} \frac{(1+t)^{d}}{\sqrt{t}} G_{\infty,0}^{(d)}(\bold{x},\bold{y};2t),\\
\label{lapkL}
\left\vert \Delta_{\bold{x}} G_{L,0}^{(d)}(\bold{x},\bold{y};t)\right\vert \leq
C_{d} \frac{(1+t)^{d}}{t} G_{\infty,0}^{(d)}(\bold{x},\bold{y};2t).
\end{gather}
\end{lema}

\begin{proof} 
From the lower bound $\sinh(\alpha) \geq \alpha$ $\forall \alpha \geq 0$ and the one in \eqref{Ek4}, \eqref{heatk} is an upper bound for \eqref{Mehler}. Hence \eqref{roughes}. \eqref{dermelh}-\eqref{lapmelh} are obtained by direct calculations. The main ingredients are:
\begin{equation}
\label{redexp}
\forall \mu, \nu>0,\, \forall x\geq 0,\quad x^{\mu} \mathrm{e}^{-\nu x} \leq \left(\frac{2\mu}{\mathrm{e} \nu}\right)^{\mu} \mathrm{e}^{-\frac{\nu}{2}x},
\end{equation}
and the following identity:
\begin{equation}
\label{Id2}
\coth(\alpha t) = \frac{1}{2} \coth\left(\frac{\alpha}{2} t\right) + \frac{1}{2} \tanh\left(\frac{\alpha}{2} t\right),\quad \forall \alpha>0,\, \forall t>0.
\end{equation}
\eqref{derkL}-\eqref{lapkL} follow from \cite[Prop. 2]{C1}.
\end{proof}

We continue with the following lemma expressing the semigroup property in the kernels sense:

\begin{lema}
\label{sempropy}
$\forall d \in \{1,2,3\}$, $\forall \delta>0$, $\forall t>0$, $\forall 0< u <t$:\\
$\mathrm{(i)}$. $\forall \kappa \geq 0$, $\forall L \in (0,\infty]$ and $\forall(\bold{x},\bold{y}) \in \Lambda_{L}^{2d}$:
\begin{equation}
\label{propsemi}
\int_{\Lambda_{L}^{d}}  G_{L,\kappa}^{(d)}(\bold{x},\bold{z}; \delta(t-u)) G_{L,\kappa}^{(d)}(\bold{z},\bold{y};\delta u)\,\mathrm{d}\bold{z} = G_{L,\kappa}^{(d)}(\bold{x},\bold{y}; \delta t).
\end{equation}
$\mathrm{(ii)}$. $\forall \kappa>0$, $\forall \gamma>0$ and $\forall(\bold{x},\bold{y}) \in \Lambda_{\infty}^{2d}$:
\begin{equation}
\label{propsemi2}
\int_{\Lambda_{\infty}^{d}}  G_{\infty,\kappa}^{(d)}(\bold{x},\bold{z}; \delta(t-u),\gamma) G_{\infty,\kappa}^{(d)}(\bold{z},\bold{y}; \delta u ,\gamma)\,\mathrm{d}\bold{z} = \gamma^{\frac{d}{2}} G_{\infty,\kappa}^{(d)}(\bold{x},\bold{y}; \delta t,\gamma).
\end{equation}
\end{lema}

\noindent \textit{Proof.} $\mathrm{(i)}$ follows from the semigroup property which reads as: $G_{L,\kappa}(t)= G_{L,\kappa}(t-u)G_{L,\kappa}(u)$ $\forall 0\leq u\leq t$. The proof of $\mathrm{(ii)}$ is based on the following explicit calculation:
\begin{multline}
\label{impor}
\forall a,b,c,d>0,\quad \int_{\mathbb{R}}  \mathrm{e}^{-\left[a(x+z)^{2}+ b(x-z)^{2}\right]} \mathrm{e}^{-\left[c(z+y)^{2}+ d(z-y)^{2}\right]}\,\mathrm{d}z = \\
\sqrt{\pi} (a+b+c+d)^{-\frac{1}{2}} \mathrm{e}^{-\frac{b(c+d)+a(d+c)+4ab}{a+b+c+d}x^{2}} \mathrm{e}^{-\frac{b(c+d)+a(d+c)+4cd}{a+b+c+d}y^{2}} \mathrm{e}^{-2 \frac{b(d-c)+a(c-d)}{a+b+c+d}xy}.
\end{multline}
Then, set $a_{0}:= \tanh(\frac{\kappa}{2} \delta u)$, $b_{0}:= \coth(\frac{\kappa}{2} \delta u)$,
$c_{0}:= \tanh(\frac{\kappa}{2}\delta(t-u))$ and $d_{0}:= \coth(\frac{\kappa}{2} \delta (t-u))$. From the  identity in \eqref{Id3}, the following one:
\begin{equation*}
\label{Id4}
\tanh(\alpha s)+ \tanh(\alpha(t-s)) = \frac{\sinh(\alpha t)}{\cosh(\alpha s) \cosh(\alpha(t-s))},\quad \forall \alpha\geq0,\, \forall t>s>0,
\end{equation*}
followed by \eqref{Id1}, one gets: $a_{0}+b_{0}+c_{0}+d_{0} = 2 \sinh(\kappa \delta t)\{\sinh(\kappa \delta u) \sinh(\kappa \delta (t-u))\}^{-1}$. The rest of the proof consists in using some identities involving the hyperbolic functions to simplify each one of the factor inside the exponentials in the r.h.s. of \eqref{impor}. It is (quite) easy to get:
\begin{gather*}
\left(b_{0}(c_{0}+d_{0})+a_{0}(d_{0}+c_{0})+4a_{0}b_{0}\right)\left(a_{0}+b_{0}+c_{0}+d_{0}\right)^{-1} = 2 \coth(\kappa \delta t),\\
\left(b_{0}(d_{0}-c_{0})+a_{0}(c_{0}-d_{0})\right)\left(a_{0}+b_{0}+c_{0}+d_{0}\right)^{-1} = \tanh\left(\frac{\kappa}{2} \delta t\right) - \coth\left(\frac{\kappa}{2} \delta t\right). \tag*{\qed}
\end{gather*}

Now, we give some estimates on the operator and trace norms of the semigroup $\{G_{L,\kappa}(t)\}_{t>0}$.
For any $\kappa\geq 0$ and $L \in (0,\infty]$, $\{G_{L,\kappa}(t)\}_{t>0}$ is a contraction semigroup, see e.g. \cite{HP}:

\begin{lema}
\label{lemNW}
$\forall d \in \{1,2,3\}$, $\forall \kappa \geq 0$ and $\forall t>0$:
\begin{equation}
\label{norm}
\forall L \in (0,\infty),\quad \left\Vert G_{L,\kappa}(t) \right\Vert \leq \left\Vert G_{\infty,\kappa}(t)\right\Vert \leq
\left(\cosh(\kappa t)\right)^{-\frac{d}{2}} \leq 1.
\end{equation}
\end{lema}

\begin{proof}  The first inequality follows from the fact that the semigroup $\{G_{L,\kappa}(t)\}_{t\geq 0}$ is increasing in $L$ in the sense of \cite[Eq. (2.39)]{BHL}. The Shur-Holmgren criterion provides the estimate on the operator norms. When $\kappa>0$, we used \eqref{impor} (with $c=0=d$) along with \eqref{Id2}.
\end{proof}

Restricting to $\kappa>0$, $\forall L\in(0,\infty]$ $\{G_{L,\kappa}(t)\}_{t>0}$ is a Gibbs semigroup (i.e. trace class, see \cite{ABN}):

\begin{lema}
\label{2trace}
$\forall d\in\{1,2,3\}$, $\forall \kappa>0$ and $\forall L \in (0,\infty]$, $\{G_{L,\kappa}(t)\}_{t>0}$ is a trace class operator on $L^{2}(\Lambda_{L}^{d})$. Moreover, denoting $E_{\infty,\kappa}^{(\bold{0})} = d \frac{\kappa}{2}$, one has for any $L \in (0,\infty)$:
\begin{equation*}
\mathrm{Tr}_{L^{2}(\Lambda_{L}^{d})} \left\{ G_{L,\kappa}(t)\right\} \leq \mathrm{Tr}_{L^{2}(\Lambda_{\infty}^{d})}\left\{G_{\infty,\kappa}(t)\right\} = \left(2 \sinh\left(\frac{\kappa}{2} t\right)\right)^{-d} = \frac{\mathrm{e}^{- E_{\infty,\kappa}^{(\bold{0})} t}}{\left(1 - \mathrm{e}^{- \kappa t}\right)^{d}}.
\end{equation*}
\end{lema}

\begin{proof} Let $(\mathfrak{I}_{2}(L^{2}(\Lambda_{L}^{d})),\Vert \cdot \Vert_{\mathfrak{I}_{2}})$ and $(\mathfrak{I}_{1}(L^{2}(\Lambda_{L}^{d})),\Vert \cdot \Vert_{\mathfrak{I}_{1}})$, $L \in (0,\infty]$ be the Banach space of Hilbert-Schmidt and trace class operators on $L^{2}(\Lambda_{L}^{d})$ respectively.
We start with $d=1$. Let $\kappa>0$ and $t>0$ be fixed. In view of \eqref{Mehler}, from \eqref{impor} (we set $c=0=d$):
\begin{equation*}
\left\Vert G_{\infty,\kappa}(t)\right\Vert_{\mathfrak{I}_{2}}^{2} = \int_{\Lambda_{\infty}^{1}}  \int_{\Lambda_{\infty}^{1}}  \left\vert G_{\infty,\kappa}^{(d=1)}(x,y;t)\right\vert^{2}\,\mathrm{d}x\,\mathrm{d}y = \frac{1}{2} \frac{1}{\sinh(\kappa t)} < \infty.
\end{equation*}
Therefore, $G_{\infty,\kappa}(t)$ is a trace class operator on $L^{2}(\Lambda_{\infty}^{1})$ since $\Vert G_{\infty,\kappa}(t)\Vert_{\mathfrak{I}_{1}} \leq \Vert G_{\infty,\kappa}(\frac{t}{2})\Vert_{\mathfrak{I}_{2}}^{2} < \infty$. Since $G_{\infty,\kappa}^{(d=1)}(\cdot\,,\cdot\,;t)$ is jointly continuous on $\Lambda_{\infty}^{2}$, from \cite[Prop. 9]{C1} it follows that:
\begin{equation}
\label{traceGk}
\left\Vert G_{\infty,\kappa}(t)\right\Vert_{\mathfrak{I}_{1}} = \int_{\Lambda_{\infty}^{1}} G_{\infty,\kappa}^{(d=1)}(x,x;t)\,\mathrm{d}x = \frac{1}{2} \frac{1}{\sinh\left(\frac{\kappa}{2} t \right)},
\end{equation}
where we used the identity \eqref{Id1}. By positivity of $G_{\infty,\kappa}(t)$, $\Vert G_{\infty,\kappa}(t)\Vert_{\mathfrak{I}_{1}}=\mathrm{Tr}_{L^{2}(\Lambda_{\infty}^{1})}\{G_{\infty,\kappa}(t)\}$.
The rest of the proof leans on the estimate \eqref{fondineq} which leads to $\Vert G_{L,\kappa}(t)\Vert_{\mathfrak{I}_{2}}^{2} \leq \Vert G_{\infty,\kappa}(t)\Vert_{\mathfrak{I}_{2}}^{2}$. Hence, $\forall L \in (0,\infty)$ $G_{L,\kappa}(t)$ is also a trace class operator on $L^{2}(\Lambda_{L}^{1})$, and by mimicking the above arguments, its trace norm obeys $\Vert G_{L,\kappa}(t) \Vert_{\mathfrak{I}_{1}} = \mathrm{Tr}_{L^{2}(\Lambda_{L}^{1})}\{ G_{L,\kappa}(t)\} \leq \Vert G_{\infty,\kappa}(t)\Vert_{\mathfrak{I}_{1}}$.
The case of $d=1$ is done. The generalization to $d=2,3$ is straightforward due to \eqref{multd}.
\end{proof}

\end{document}